\documentclass[11pt]{article} 
\usepackage[utf8]{inputenc} 
\RequirePackage{amsthm,amsmath,amsfonts,amssymb,mathtools} 
\RequirePackage{graphicx} 
\RequirePackage[square,numbers]{natbib} 

\usepackage{authblk} 
\usepackage[dvipsnames]{xcolor} 
\usepackage[plainpages=false,pdfpagelabels, colorlinks=true,linkcolor=RedOrange, urlcolor=MidnightBlue,citecolor=MidnightBlue]{hyperref}
\usepackage{MnSymbol}
\usepackage{caption}
\usepackage{subcaption}
\usepackage{booktabs}
\providecommand{\keywords}[1]{\small \textbf{\textit{Keywords---}} #1}

\usepackage{fullpage}
\usepackage{caption}
\usepackage{subcaption}
\usepackage{my-macros}

\renewcommand{\ast}{\oast}
\newcommand{\eq}[1]{\eqref{eq:#1}}
\newcommand{\transp}{\top}
\newcommand{\reallywidehat}{\widehat}

\title{Randomization Inference When N Equals One}

\begin{document}

\author[1]{Tengyuan Liang} 
\author[2]{Benjamin Recht}

\affil[1]{University of Chicago} \affil[2]{University of California, Berkeley}
\date{}

\maketitle 
\begin{abstract}
	N-of-1 experiments, where a unit serves as its own control and treatment in different time windows, have been used in certain medical contexts for decades. However, due to effects that accumulate over long time windows and interventions that have complex evolution, a lack of robust inference tools has limited the widespread applicability of such N-of-1 designs. This work combines techniques from experiment design in causal inference and system identification from control theory to provide such an inference framework. We derive a model of the dynamic interference effect that arises in linear time-invariant dynamical systems. We show that a family of causal estimands analogous to those studied in potential outcomes are estimable via a standard estimator derived from the method of moments. We derive formulae for higher moments of this estimator and describe conditions under which N-of-1 designs may provide faster ways to estimate the effects of interventions in dynamical systems. We also provide conditions under which our estimator is asymptotically normal and derive valid confidence intervals for this setting.
\end{abstract}


%%%%%%%%%%%%%%%%%%%%%%%%%%%%%%%%%%%%%%%%%%%%%%%%%%%%%%%%%%%%%%%%%%%%%%%%%%%%%%%%%%%%%%%%%%%%
%%%%%%%%%%%%%%%%%%%%%%%%%%%%%%%%%%%%%%%%%%%%%%%%%%%%%%%%%%%%%%%%%%%%%%%%%%%%%%%%%%%%%%%%%%%%
%%%%%%%%%%%%%%%%%%%%%%%%%%%%%%%%%%%%%%%%%%%%%%
%% Main Paper                               %%
%%%%%%%%%%%%%%%%%%%%%%%%%%%%%%%%%%%%%%%%%%%%%%
% \tableofcontents

\section{Introduction}

Randomized experiments are arguably the most important contribution of twentieth-century statistics, providing a program to establish the value of interventions to populations of individuals. By randomly assigning individuals to treatment and control, statistical analysis can simultaneously determine the magnitude of intervention effects and eliminate confounding explanations. However, a significant drawback of randomized experiments is that they yield conclusions about populations rather than individuals. When treatments present heterogeneous effects, a randomized experiment cannot inform individuals whether a treatment can work for them.

Is it possible to design experiments for individuals? \emph{N-of-1 trials} attempt to solve this problem by having individuals apply randomized treatments in different time windows. The individual becomes the treatment and the control. N-of-1 trials are within-patient randomized controlled trials, where the treatment is randomized at every time period, with the subject potentially crossing over between treatment and control at each phase. 

In medicine, the formalized N-of-1 trial is nearly as old as the randomized clinical trial itself. In 1950, only a few years after the famous randomized trials of Streptomycin, Quin and colleagues performed randomized within subject trials of an arthritis reduction agent where patients received a random treatment at multiple doctor visits~\citep{quin_clinical_1950}. D. D. Reid explained Latin Square designs for N-of-1 trials and explored these for testing the effectiveness of laxatives~\citep{Reid60} and arthritis ointments~\citep{ReidLancet54}. Snell and Armitage tested heroin as a cough suppressant with N-of-1 designs~\citep{snell_clinical_1957}. Bradford Hill even devotes a lengthy part of the Clinical Trials chapter (pages 260-263) in his \emph{Principles of Medical Statistics} to the design of N-of-1 trials~\citep{HillPrinciplesBook}. 

Since the 1950s, though never a dominant experimental design, N-of-1 trials have continued to find application. Louis et al. reported 28 self-controlled studies in the New England Journal of Medicine in 1978 and 1979~\citep{louis_crossover_1984}. Guyatt et al. argued for using N-of-1 trials as part of general clinical practice with patients, helping individuals find the best treatment for them~\citep{guyatt_determining_1986}. N-of-1 trials have been recently used successfully to test the side effects of statin treatments~\citep{wood_n--1_2020}.

% In the context of online controlled experiments and AB tests.
In the online experimentation community, N-of-1 experiments are often called switchback or interleaving designs. Netflix and LinkedIn use N-of-1 designs to test the time-varying effects of their recommendation and matching algorithms. For example, \cite{bojinov2020avoid} reports interleaving online experiments where a client is randomly exposed to different app experiences for 30-minute intervals. N-of-1 designs have the potential to be tailored to personalized experiments to improve the experience of individual customers. Moreover, N-of-1 designs enable data scientists to mitigate interference effects from multiple AB tests running concurrently~\citep{thomke2020experimentation}.

If the effect is immediate and transient, standard tools from randomized experiments can be used to estimate its magnitude and significance in N-of-1 trials. However, when the effect of an intervention persists over time, it becomes challenging to draw proper inferences. Though recent estimators have attempted otherwise, most N-of-1 designs try to reduce the problem to a randomized experiment where each time period is independent. In such a case, the treatment must take effect and cease rapidly. These designs also must assume an individual's baseline is roughly constant. For these reasons, N-of-1 designs have primarily focused on chronic conditions treatable by acute medications.

Can we do better? In this work, we aim to combine techniques from experiment design in causal inference with system identification from control theory to provide a general inference framework for N-of-1 experiments. In what follows, we describe the foundational pillars of our framework and a view of how they might be combined for inferring effects in N-of-1 experiments. We then present a formal interference model for temporal effects drawn from the literature on dynamical systems. We define a family of estimands of causal effects of these models and then present a reasonable, unbiased estimator for this family. We demonstrate that this estimator is asymptotically normal and provide a plug-in estimator for its variance. Together, these enable the construction of valid confidence intervals. Along the way, we discuss the inferential challenges of this framework and why our approach, though more general than previous methods, remains limited.

\subsection{Foundational Related Work}
Our work draws and expands upon the literature on causal inference and system identification. In this section, we elaborate on the foundations we draw upon and how we will attempt to reconcile these two lines of research. There is also a growing list of existing works to study dynamic treatment effects. To better position our work, we will discuss and compare the existing approaches in detail in Section~\ref{sec:comparison}, after introducing our model for temporal interference, estimands, and estimators.

\paragraph{Causal inference} The potential outcomes framework \citep{neyman1923application,rubin1974estimating} offers a systematic and rigorous approach to studying causal relationships in experimental and observational studies. In the simplest setting, a \emph{unit} $i$ has two potential outcomes under an intervention (also called a treatment). The observed outcome of unit $i$ is denoted by $y_i$, and we distinguish the outcome under no treatment as $y_i(0)$ and under treatment as $y_i(1)$. If we let $x_i$ denote the binary indicator equal to $1$ when treatment is applied and $0$ when treatment is not applied, then we have the expression
$$
	y_i := y_i(1)\cdot x_i + y_i(0) \cdot(1-x_i)
$$
This formula holds regardless of whether $y_i$ are real-valued or binary outcomes.

Though this formula looks like it is tautological, it makes some hidden assumptions. In particular, this model assumes that the treatment applied to unit $i$ only impacts unit $i$. This is called the \emph{Stable Unit Treatment Values Assumption} (SUTVA). This assumption does not hold for within-unit trials. If a treatment is applied in succession to a unit, care must be taken to ensure that the previous treatment does not influence the outcomes of future treatments.

When the $x_i$ can be intentionally assigned, random assignment enables estimation of the average outcome under treatment and control. That is, if we are interested in estimating the quantity
\begin{align*}
	\tau := \tfrac{1}{n} \sum_{i=1}^n y_i(1)-y_i(0) \;,
\end{align*}
we can do an experiment where $x_i$ is assigned at random and then estimate $\tau$ from the observed outcomes.

Assume, for simplicity, that the $x_i$ are assigned as independent Bernoulli coin flips. Denote the centered treatment variable $z_i = 2x_i-1 \in \{\pm 1\}$. Then the \emph{Horvitz-Thompson estimator} $\widehat{\tau}_{\mathrm{HT}}$ 
\begin{align*}
	\widehat{\tau}_{\mathrm{HT}} := \tfrac{2}{n} \sum_{i=1}^n z_i \cdot y_i \;
\end{align*}
is an unbiased estimator of the average treatment effect. When SUTVA applies, inference about the significance of causal effects is also elegant. $\widehat{\tau}_{\mathrm{HT}}$ is asymptotically normal, and confidence intervals can be constructed using upper bounds on the variance of $\widehat{\tau}_{\mathrm{HT}}$ (see, for example,~\cite{li2016exact}).

\paragraph{Beyond SUTVA} For causal questions about time series data, SUTVA rarely applies~\citep{sims1972money,granger1969investigating}. Suppose each unit is exposed to a sequence of interventions indexed by time, $x_1$, $x_2, \ldots, x_T$. Even if these interventions are binary-valued, the outcome at time $t$, $y_t$ could potentially depend on all of the interventions (or outcomes) before time $t$. Namely, the potential outcomes $y_t(x_1,\ldots, x_t)$ is a function of $x_1,\ldots, x_t$. We could equivalently represent the observed outcome $y_t$ then as a sum over all possible interactions---a discrete Fourier expansion of $y_t(x_0, x_1, \ldots, x_t)$ over the hypercube, see \cite{o2014analysis}. That is, without loss of generality, 
\begin{align}
	\label{eqn:discrete-fourier}
	y_t = \sum_{S: S\subseteq \{0, 1,\ldots, t\}} \alpha_{t, S} \cdot \prod_{i \in S} z_i\;,
\end{align}
where again, $z_i=2x_i-1$. The standard Potential Outcomes framework is the special case when the coefficients $\alpha_{t, S}  = 0$ for all $S \neq \emptyset$ or $\{ t\}$. However, the expansion has a number of terms that grow exponentially with $t$, a curse of dimensionality preventing meaningful statistical analysis.

Since a full representation is intractable, more restricted interactions have been studied. The Granger-Sims causality framework \citep{granger1969investigating,sims1972money,angrist2011causal} tests correlations in $x$'s that can explain $y$. For example, in its simplest form, Granger Sims studies how well linear expansions of the form
\begin{align*}
	y_t = \alpha_{\emptyset} + \alpha_0 \cdot x_t + \alpha_1 \cdot x_{t-1}+ \ldots + \alpha_{t} \cdot x_{0} + \mathrm{error}_t 
\end{align*}
predict outcomes. Granger causality assumes \emph{time-invariance} to provide practically useful answers to understand when time series can predict each other.

In this paper, we study this time-invariant model, under the assumption that the $x_i$ are chosen interventionally as part of an experiment.

\paragraph{System identification} In control theory, modeling the input-output behavior of a dynamical system is a first step in designing a feedback policy. Such dynamical systems are assumed to obey dynamical equations of the form
$$
	s_{t+1} = f_t(s_t, x_t,\varepsilon_t), ~~
	y_t = h_t(s_t,x_t,\varepsilon_t)\,.
$$
$x_t$ in this context is a controllable input to the dynamical system. $y_t$ is the measured output and $\varepsilon_t$ is an exogenous noise process. $s_t$ is the \emph{state} of the dynamical system. The state is the variable that makes future outputs conditionally independent of past outputs.

The goal in system identification is to find appropriate estimates of the function $f_t$ and $h_t$ so that inputs $x_t$ can be planned to steer $y_t$ to desired values. 

Just as we have described above, the identification problem is only tractable if the class of models for $f_t$ and $h_t$ is simple enough. Several reduced model classes are useful for representing dynamical systems. One of the most fundamental models assumes that the functions $f_t$ and $h_t$ are linear (but might vary over time). In such a case, we can write the input-output map as
$$
	y_t = \sum_{k \leq t} G_{tk} x_k + e_t
$$
where $G_{tk}$ are scalars and $e_t$ are linear functions of the $\varepsilon_k$ variables and $s_1$.

Under a further restriction, we can assume the system is linear and \emph{time-invariant} so that $f_t=f$ and $h_t=h$ for all $t$. In this case, $y_t$ can be written as a linear convolution of $x_k$ and a sequence $g_k$:
$$
	y_t = \sum_{k \leq t} g_{t-k} x_k + e_t\,.
$$
This model is the same as the 1-step Granger Causality model described in the previous section.

Identifying linear time-invariant systems has a rich literature (see the textbook by Ljung~\citep{LjungBook}). Importantly, the choice of $x_k$ matters as this input design must make all $g_k$ of interest identifiable. The most popular choice in theory and practice chooses $x_k$ to be an i.i.d. sequence of zero mean random variables. In practice, it suffices that these variables take on two values, such as $1$ and $-1$~\citep{mareels1984sufficiency,Verhaegen92,Overschee94}. When restricted to such signed inputs, the linear system identification problem looks like a Potential Outcomes estimation problem under a linear interference model.

Though there is extensive literature on system identification, only recently have researchers determined statistical error rates for the estimation of the sequence $g_k$ from a single input sequence $x_k$. There are numerous reasons for this. First, since system identification of engineered systems is typically done in a controlled setting, sample complexity is not typically a constraint on engineers. Second, the statistical methods for system identification lean on random matrix results that have only been recently derived.

The analysis of \cite{oymak2019non} shows that the parameters $g_k$ can be estimated via least squares when $x_t$ and $e_t$ are zero-mean Gaussian random variables and $s_1=0$. They further assume that $g_k$ obeys a decay condition that $|g_k| \leq C a^k$ for a constant $C$ and $a<1$. They show that the rate of error in the estimate $\hat{g}_k$ satisfies
$$
	\sum_{k=1}^\infty |g_k- \hat{g}_k|^2 = O\left(\tfrac{1}{(1-a)^2 T}\right)
$$
where $T$ is the number of observed samples.

Later work by \cite{Simchowitz19} removed the dependence on $a$ in exchange for another property of $g$ called \emph{phase rank}, but some restriction on the form of $g$ is always required to estimate $g$. Similarly, work by \cite{bakshi2023} provides tighter bounds on this estimation error, but assumes the error signal is random and the input signal is zero mean. Moreover, the current system identification literature does not provide provable means of data-driven uncertainty quantification, a crucial component of causal inference.

\paragraph{A synthesis of causal inference and system identification} The work in this paper builds on these three pillars to synthesize an inference framework for N-of-1 experiments. First, as we develop in the next section, we focus on the linear time-invariant outcome model studied in Granger-Sims causality, and in linear system identification
\begin{align}\label{eq:basic-model}
	y_t = \sum_{s=0}^{t} x_{s} g_{t-s} + e_t\,.
\end{align}

This model generalizes the Neyman-Rubin framework to allow for a particular functional form of linear interference that is inspired by linear dynamical systems. In the next section, we develop a generalization of the Horvitz-Thompson estimator for 
linear estimands of the form
\begin{align*}
	\tau(q) := \langle g, q \rangle \,.
\end{align*}
In the case that $q$ is the first standard basis vector, the generalized estimator will precisely be the Horvitz-Thompson estimator.

Unlike in the Granger-Sims and linear system identification frameworks, we do not assume the error signal $e_t$ is random in~\eqref{eq:basic-model}.  We will focus on estimating various linear functionals of the coefficients $g \in \mathbb{R}^T$. Focusing on linear estimands rather than the coefficients themselves should simplify the problem. We note, however, that results for this simplification cannot be derived from the prior art.

Additionally, in this work, we are interested in estimating $g_k$ when $e_t$ is a non-random signal. 
%The only prior work that attempted to estimate $g$ when $e_t$ was non-random was the work of Simchowitz et al~\cite{Simchowitz19}. 
\cite{Simchowitz19} was the first to study estimation of $g$ when $e_t$ was non-random, and estimation under non-random conditions has been used in the study of online control~\citep{hazan2020nonstochastic}. However, none of these works provided specific forms of the error, nor did they quantify the variance sufficiently accurately for the construction of asymptotically-sharp confidence intervals.
 
However, the explicit form of the error is not provided. There are several practically relevant situations where such non-random noise arises. For instance, suppose we restrict our attention to input sequences that are constrained to be nonnegative. This departs from the zero-mean inputs in Granger-Sims and system identification. Non-zero-mean inputs can be modeled as zero-mean inputs plus a new, non-random error term.

Unlike system identification, our work is interested in estimating confidence intervals on linear functionals of $g_k$ from finite experiments. None of the prior work has investigated such uncertainty quantification. In particular, we derive novel asymptotic normality guarantees for our estimator, which is closely related to the least-squares estimator of \cite{oymak2019non} and the method-of-moments estimator of \cite{bakshi2023}. The proof of asymptotic normality is more delicate than what is typical of Horvitz-Thompson estimators. In particular, the asymptotic normality requires an intricate calculation balancing of eighth moments of the intervention variables.

\section{Treatment Effects with Temporal Interference: Estimands and Estimators}
\label{sec:linear-conv}

\subsection{Convolutions}
\label{sec:convolution}
The outcomes in linear time-invariant models can be represented as a convolution between an intervention sequence and a parameter vector known as an \emph{impulse response}. This impulse response also quantifies the temporal interference. Before expanding on our statistical theory, we first review the basic notation and theory of convolution. This notation will simplify many of the formulas, and eliminate the need to track subscripts.

Let $u_t$ and $v_t$ be two signals with indices in the nonnegative integers. The convolution of $u_t$ and $v_t$ is the signal
$$
	(u * v)_t = \sum_{s=0}^t u_s v_{t-s}\;.
$$
Note how this expression allows us to compactly rewrite the models of Granger-Sims causality and linear time-invariant systems.

Convolution is linear and commutative, namely $u * v = v * u$. We can write the operator that takes $v$ to $u * v$ in matrix form
$$
	 u*v = \mathcal{T}_u v
$$
where $\mathcal{T}_u$ is the Toeplitz matrix
$$
	\mathcal{T}_u = \begin{bmatrix}
		u_0 & 0 & 0 & 0 & \ldots & 0\\
		u_1 & u_0 & 0 & 0 & \ldots & 0\\
		u_2 & u_1 & u_0 & 0 & \ldots & 0\\
		\vdots & \vdots & \vdots & \ddots &\vdots & \vdots\\
		u_{T-1} & u_{T-2} & u_{T-3} & \ldots & \ldots & u_0
	\end{bmatrix}\,.
$$
Note that we must have $\mathcal{T}_u v = \mathcal{T}_v u$.

The adjoint operator of $\mathcal{T}$---that satisfies $\langle \mathcal{T}^*(M), u \rangle = \langle M, \mathcal{T}_u \rangle$ for all $M, u$---maps matrices to vectors with components
$$
	\mathcal{T}^*(M)_s = \sum_{t = s}^{T-1} M_{t,t-s}, ~ \forall s = 0, \ldots, T-1 \;.
$$
Moreover, we have the composition $\mathcal{T}^*(\mathcal{T}_u)_s =  (T-s) u_s$.

In our theoretical analysis, we use \emph{circular extensions} of signals that result in tractable closed-form expressions for many of the moments in our experiment designs. 
For two signals $u$ and $v$ of length $T$, the \emph{circular convolution} is defined as
$$
	(u \ast v)_t = \sum_{s=0}^{T-1} u_s v_{t-s\pmod{T}} \;,
$$
where the circular extention notation $\pmod{T}$ means that for $s\in [0, t], t-s\pmod{T} = t-s$ and for $s\in (t, T-1], t-s\pmod{T} = T+t-s$.
Note here that the summation extends over all possible indices $s$, not just those that have values less than or equal to $t$. In matrix form, circular convolution corresponds to multiplication by a circulant matrix
$$
	 u \ast v = \mathcal{C}_u v
$$
where $\mathcal{C}_u$ is the matrix
$$
	\mathcal{C}_u = \begin{bmatrix}
		u_0 & u_{T-1} & u_{T-2} & u_{T-3} & \ldots & u_1\\
		u_1 & u_0 & u_{T-1} & u_{T-2} & \ldots & u_2\\
		u_2 & u_1 & u_0 & u_{T-1} & \ldots & u_3\\
		\vdots & \vdots & \vdots & \ddots &\vdots & \vdots\\
		u_{T-1} & u_{T-2} & u_{T-3} & \ldots & \ldots & u_0
	\end{bmatrix}\,.
$$
For a variety of algebraic and Fourier analytical reasons, computation with circular convolutions is more elegant than with linear convolutions. For example, the Discrete Fourier Transform of a circular convolution is the product of the Discrete Fourier Transforms of the individual signals. Such properties will be used in depth in the theoretical analysis.

For the ease of writing subscripts in equations, we introduce the circular extension notation $v^{\circ}_t = v_{t\pmod{T}}$ for possible negative integer $t$. With this notation, 
\begin{align}
	\label{eqn:circ-abbrev}
	( u \ast v )_t &= \sum_{s=0}^{T-1} u_s v^{\circ}_{t-s} \;. 
\end{align}

\subsection{Estimands}

Suppose we have a scalar outcome sequence $y_t$ that depends on a scalar treatment $x_t$ via a linear dynamical system
\begin{equation}\label{eq:lti-dynamics}
	y_t = \sum_{s=0}^{t} g_s x_{t-s} + e_t = (g * x)_t + e_t\,.
\end{equation}
Here $g$ is the \emph{impulse response} of the linear system. For the simplicity of this paper, we only consider that the $x_t$'s are binary-valued, namely Bernoulli N-of-1 experiments, but we expect extensions to real-valued inputs to be relatively straightforward.

$e_t$ is an exogenous signal that is not affected by the treatment $x_t$. But we are interested in analyzing $e_t$ that are far from random. As a motivating example, suppose two treatments have treatment effects $g^{(A)}$ and $g^{(B)}$. We are interested in which leads to overall better outcomes. In this case, $g = g^{(A)}-g^{(B)}$ and $e=g^{(B)}$. These signals are tightly correlated. Hence, we model $e$ as an \emph{adversarial error} oblivious to the randomization $x$. The error could be non-i.i.d., non-stationary, and even systematic, as long as it is oblivious to input $x$. 

Which outcome properties might we be interested in? One example is to study the \emph{average treatment effect}, which compares outcomes when $x_t$ is equal to either all ones or all zeros, denoted by $\mathbf{1}$ and $\mathbf{0}$ respectively:
\begin{align*}
	\tau_{\mathrm{ATE}} = \tfrac{1}{T} \sum_{t=0}^{T-1} y_t(\mathbf{1})-y_t(\mathbf{0})
\end{align*}
Using~\eq{lti-dynamics}, we can write $\tau_{\mathrm{ATE}}$ in terms of the impulse response
\begin{align*}
 \tau_{\mathrm{ATE}} = \sum_{t=0}^{T-1}  \tfrac{T-t}{T} g_{t}\,.
\end{align*}
Hence, for full generality, we will study linear estimands of the form $\tau^{\mathcal{L}}(q) = \langle D_T q, g \rangle$ where $q$ is a vector and $D_T$ is the diagonal matrix with $t$-th entry $\frac{T-t}{T}$. With this definition, $\tau_{\mathrm{ATE}}=	\tau^{\mathcal{L}}(q)$ when $q=\mathbf{1}$. These different estimands index by $q$ can highlight varied intervention effects, and we shall see a few other examples in Section~\ref{sec:example-utility}.

In a standard N-of-1 trial, the estimand would compare the effects of treatments at a time resolution of each dosing period. That is, if a particular treatment is given for a week at a time, the outcomes will be measured at the resolution of each week. Standard N-of-1 designs also use long periods to ensure that the interference between the effects of the different treatments have washed out before an outcome is measured. If a treatment has spillover, then measurements early in a treatment period will be influenced by the prior period's treatment.

But what if we are interested in the average treatment effect at the resolution of a day?
Let us consider for example the following two treatments with a two-day period:
\begin{align*}
	&y_{1}^{(A)} = x_1\,, ~~~	y_{2}^{(A)} = x_1 + x_2 \;,  \\
	&\text{and} ~~~ y_{1}^{(B)} = 0 \,, ~~~y_{2}^{(B)} = 2 x_1\,.
\end{align*}
Both of these treatments correspond to linear time-invariant systems with $g^{(A)} = (1,1)$ and $g^{(B)} = (0,2)$. For both treatments, the outcome if the treatment is taken for a two-day period is $2$. Therefore, if the dosing period is two days, Treatment A and Treatment B are indistinguishable in standard N-of-1 trials. However, the average outcome for Treatment A is $1.5$, and the average outcome for Treatment B is $1$. If, for example, these treatments are pain medications, then Treatment A is preferable to Treatment B as the immediate relief is valuable. We thus aim to allow for a variety of linear estimands to capture the effect most aligned with preferred outcomes.

\subsection{Estimation from Random Design}

Consider the randomized experiment where we assign the input at time $t$ to be zero or one with equal probability, independently at each $t \in [T]$. How can we estimate the treatment effect $\tau^{\mathcal{L}}(q)$ from the observed outcomes of an experiment? We care about the linear functionals $\tau^{\mathcal{L}}(q)$ as estimands as they compare any two counterfactual treatment paths as a special case. For $q, q' \in \{ 0, 1 \}^T$ on the hypercube, $\tau^{\mathcal{L}}(q) -\tau^{\mathcal{L}}(q')$ evaluates the difference between any two counterfactual paths, an exemplar question in causal inference.

We now fix the notation for the remainder of the document to distinguish between random and nonrandom signals. We will use boldface font for random vectors to emphasize the stochastic nature of experimentation. In what follows, we use the normal font, say $x \in \mathbb{R}^{T}$ and $H \in \mathbb{R}^{T \times T}$ to denote deterministic vectors and matrices, $x_i, H_{ij}$ to denote their corresponding entries. We use the boldface, say $\bx, \by, \bz, \bW$, to denote vectors that are random due to experimental design. The lone exception is that we use $\mathbf{1}$ to denote the all ones vector and avoid confusion with the scalar $1$. When there is no ambiguity, we use the abbreviation $[T] := \{ 0, 1, \ldots, T-1\}$ for simplicity of the time index.

With our fixed notation, the linear convolution model is written as
\begin{align}
	\label{eqn:linear-conv-model}
	\by = \bx * g + e \in \mathbb{R}^T \;,
\end{align}
and $g, \bx, \bx * g$ and $e$ are all of length $T$ with $\bx_t \in \{0,1\}, t\in [T]$. 

A natural generalization of the Horvitz-Thompson estimator is the method-of-moments estimator, defined as
\begin{align}\label{eq:q-estimator}
	\widehat{\tau}^{\mathcal{L}}(q) = \tfrac{1}{T} \langle (2\bx-\mathbf{1})*q, 2 \by \rangle\,.
\end{align}
This choice is natural in the following sense. Suppose we restrict our attention to unbiased linear estimators $\langle w, \by \rangle$. Then unbiasedness is the linear constraint
$
	\E_{\bx}[\langle w, \by \rangle] = \langle D_T q,g \rangle .
$
Setting $w$ as the Moore-Penrose pseudoinverse of this system of equations yields the estimator~\eq{q-estimator}.

As we will do several times in what follows, we note that $\widehat{\tau}^{\mathcal{L}}(q)$ is a quadratic function of $T$ independent Rademacher random variables. Throughout we will need to compute moments of such Rademacher chaos, and we provide explicit formulas for their first, second, and fourth moments in Appendix~\ref{app:chaos}.

Indeed,  $\bz=2\bx-\mathbf{1}$ is a vector of independent Rademacher random variables and we have that
\begin{align*}
	\widehat{\tau}^{\mathcal{L}}(q) &= \tfrac{1}{T} \langle \bz * q, 2(\bx * g + e) \rangle
	= \tfrac{1}{T} \langle \mathcal{T}_q \bz, 2\mathcal{T}_g \bx+ 2e\rangle 
	%= \langle 2z,  \mathcal{T}_q^\transp\mathcal{T}_g x+\mathcal{T}_q^\transp e\rangle \\
	= \tfrac{1}{T} \left\langle \bz,   \mathcal{T}_q^\transp \mathcal{T}_g \bz\right\rangle
	+\tfrac{1}{T}  \left\langle \bz, \mathcal{T}_q^\transp\mathcal{T}_g \mathbf{1} + 2\mathcal{T}_q^\transp e\right\rangle \;.
\end{align*}
From this calculation we can verify that $\widehat{\tau}^{\mathcal{L}}(q)$ is indeed unbiased as
$$
	\E[\widehat{\tau}^{\mathcal{L}}(q)] =\tfrac{1}{T}  \trace\left( \mathcal{T}_q^\transp \mathcal{T}_g \right) = \sum_{t \in [T]} \tfrac{T-t}{T}  q_t g_t = \langle D_T q, g \rangle \;.
$$

Moreover, the variance of $\hat{\tau}_q$ is given by
\begin{equation}\label{eq:mom-variance}
	 T^2 \operatorname{var}(\widehat{\tau}^{\mathcal{L}}(q)) = 
	\|\mathcal{T}_q^\transp \mathcal{T}_g\|_F^2 + \trace(\mathcal{T}_q^\transp \mathcal{T}_g\mathcal{T}_q^\transp \mathcal{T}_g) - 2 
	\sum_{t\in [T]} (\mathcal{T}_q^\transp \mathcal{T}_g)_{tt}^2 +  \| \mathcal{T}_q^\transp (g * \mathbf{1}+2 e)\|^2	 \;.
\end{equation}
Though this variance formula can be computed from the problem data, it is not wieldy or illuminating. In the later sections, we will compute an approximate variance formula that provides more insights into the dependence of the variance on $g$, $q$, and $e$.

\subsection{Example Utility of N-of-1 Design}
\label{sec:example-utility}
The simple variance calculations thus far have already shed some light on the utility of rapid switching in N-of-1 designs. In this section, we describe a synthetic example illustrating how rapid switching can lead to approximately the same inference time, but lead to better patient outcomes and provide more information than standard experiment designs.

Suppose a patient is trying to choose between two drugs, Drug A and Drug B. If one wants to use a classical treatment effect estimator, such as the Horvitz-Thompson estimator, the period at which the treatment is switched must be carefully chosen. We have to ensure that the effects of one drug don't interfere with the other in order to get a precise estimate of the treatment effect. In our N-of-1 designs, the drugs can be interleaved more rapidly. The benefit of such an approach is that if there are decaying effects, the patient will experience the average response of the two treatments. But this averaging may come at the cost of running a longer experiment. Here, we showcase a toy example illustrating how the expected time required to distinguish two treatments is approximately the same. Yet, the average outcome is better in our N-of-1 design with rapid interleavings.

Suppose that Drug A and Drug B have impulse responses
$$
	g_t^{(A)} = A g_t, \quad  g_t^{(B)} = B g_t \;.
$$
Let $A$ and $B$ be positive scalars and, without loss of generality, $A\geq B$. As an illustrative example, suppose that $g_t = \beta^t$ has a simple exponential decay. Suppose $t$ indexes days and $\beta=0.5$. This means that the treatment has a half-life of one day. The half-life of the drugs determines the switching time. For an N-of-1 experiment that requires the effect to clear, this means you should wait 4-5 days between potentially switching treatments. Let us choose $\ell = 5$ to wash out interference.

Consider three simple estimands, all special cases of the general estimands studied in this paper
\begin{align*}
	\textit{immediate effect}:~ &\tau^{\textrm{imd}}:= g_0^{(A)} -g_0^{(B)} \;, \\
	\textit{cumulative effect}:~ &\tau^{\textrm{cum}}:= \big(g_0^{(A)} + g_1^{(A)} \big) - \big(g_0^{(B)} + g_1^{(B)} \big)\;, \\
	\textit{flip effect}:~ &\tau^{\textrm{flp}}:= \big(g_0^{(A)} + g_1^{(B)} \big) - \big(g_0^{(B)} + g_1^{(A)} \big) \;.
\end{align*}
The \textit{immediate effect} compares treatment paths $(A, -)$ vs. $(B, -)$, which equals $(A-B)g_0$; the \textit{cumulative effect} compares paths $(A, A, -)$ vs. $(B, B, -)$, namely $(A-B)(g_0 + g_1)$; the \textit{flip effect} compares paths $(B, A, -)$ vs. $(A, B, -)$, that is $(A-B)(g_0 - g_1)$. These three estimands are special cases of the general estimands $\tau(q)$ we study, with $q = (1, 0,\ldots, 0)$ for the \textit{immediate effect}, $q = (1, 1, 0, \ldots, 0)$ for the \textit{cumulative effect}, and $\tau(q) - \tau(q')$ for the \textit{flip effect} where $q=(0, 1, 0, \ldots, 0)$ and $q' = (1, 0, 0, \ldots, 0)$.

Let us compare three N-of-1 designs: (1) \textit{Standard Treatment imd}, a standard experiment design where we treat and measure the immediate outcome, wait $\ell$ days for the transient dynamics to end, and then potentially switch the treatment. We choose $\ell = 5$ in the numerical comparison. We note this design is specifically for measuring the \textit{immediate effect}.  (2) \textit{Standard Treatment cum}, another standard N-of-1 design where we treat either Drug A or B for a consecutive period of 2 days, then measure the cumulative outcome after day two. Like before, we take a gap of $\ell$ days for transient dynamics to end before the switch. This design is suited for measuring the \textit{cumulative effect}. (3) \textit{Our Rapid Interleaving}, our interleaving design where we treat daily, measure the effect, and switch the treatment randomly. Figure~\ref{fig:A_vs_B_sample_path} illustrates the impulse response curves for Drugs A and B, as well as a sample path of treatment sequence and response for three designs.

\begin{figure}[!htbp]
    \caption{\scriptsize From top to bottom row: \textit{Standard Treatment imd}, \textit{Standard Treatment cum}, and \textit{Our Rapid Interleaving}.}
    \label{fig:A_vs_B_compare_designs}
     \centering
     \begin{subfigure}[t]{0.49\textwidth}
         \centering
         \includegraphics[height=4.9cm]{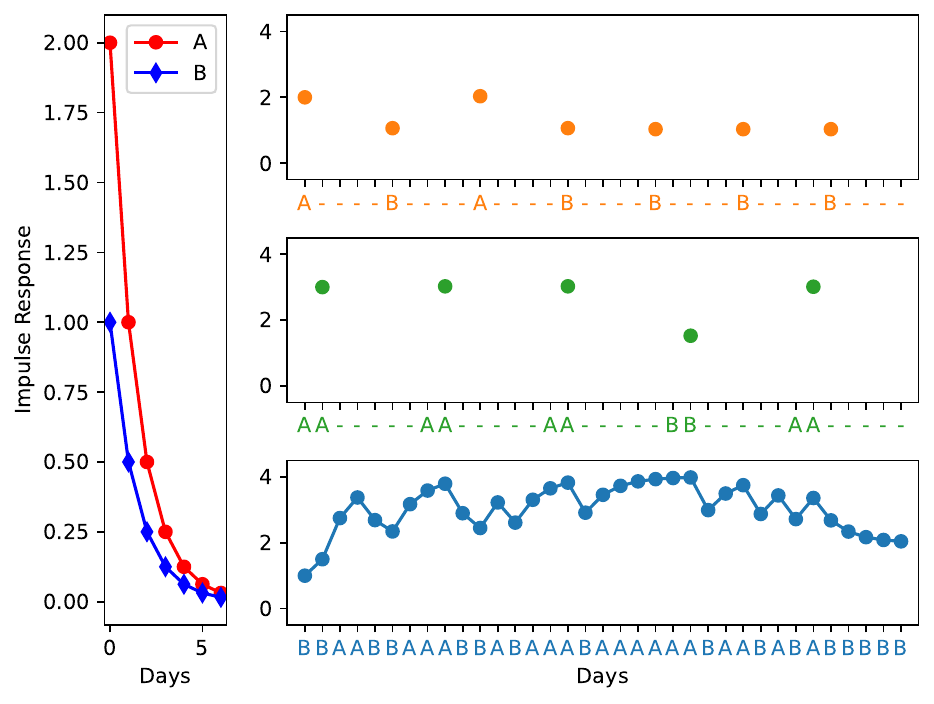}
         \caption{\scriptsize Sample path of treatments (x-tick labels) and responses (y-axis). Each row corresponds to a treatment plan; from top to bottom, each row corresponds to \textit{Standard Treatment imd}, \textit{Standard Treatment cum}, and \textit{Our Rapid Interleaving}.}
         \label{fig:A_vs_B_sample_path}
     \end{subfigure}%
     \hfill
     \begin{subfigure}[t]{0.49\textwidth}
         \centering
         \includegraphics[height=4.9cm]{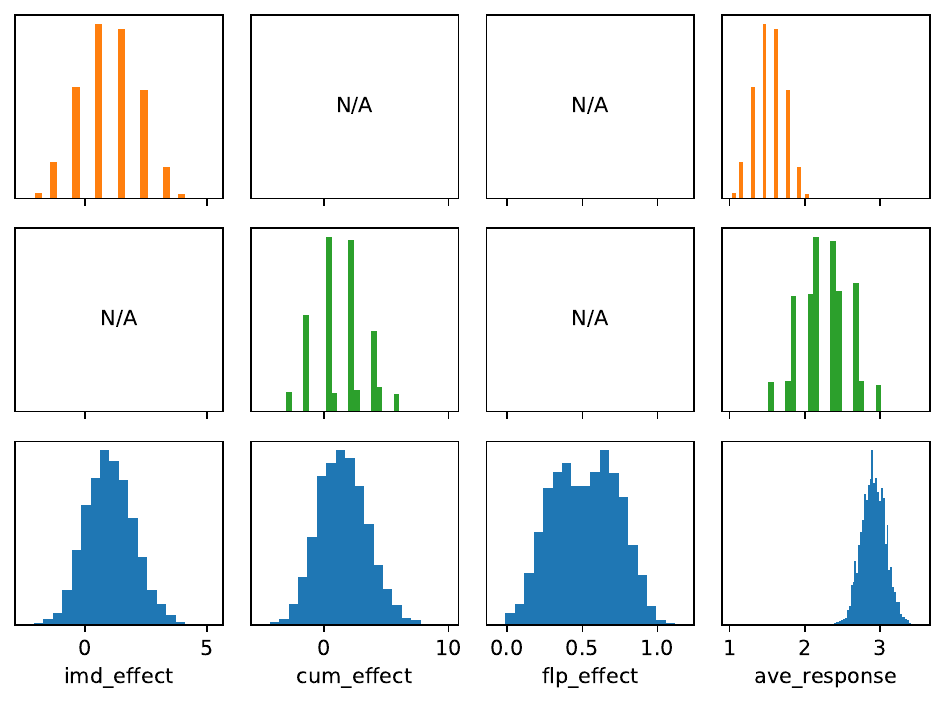}
         \caption{\scriptsize Histogram across simulations: the left three columns plots $\widehat{\tau}$, for \textit{immediate}, \textit{cumulative}, and \textit{flip} effects; the last column plots average response $\tfrac{1}{T} \sum_t y_t$. A histogram will be left as ``N/A'' when the specific design cannot be used to estimate a certain effect.}
         \label{fig:A_vs_B_snr}
     \end{subfigure}
\end{figure}

\begin{table}[!htbp]
	\caption{Example: $A, B = 2, 1$ and $g_t = 0.5^t$ for a period of $T = 35$ days. Each column, row, and cell correspond to the same simulations as in Figure~\ref{fig:A_vs_B_snr}.}
	\label{tab:A_vs_B_compare_designs}
	\centering
\resizebox{\textwidth}{!}{
\begin{tabular}{l c | c c | c c | c c | c}  
\toprule
N-of-1 Design	& Any $\tau(q)$?  & \multicolumn{2}{c|}{\textit{immediate effect}} & \multicolumn{2}{c|}{\textit{cumulative effect}} & \multicolumn{2}{c|}{\textit{flip effect}} & \textit{average response} \\
	&   & \multicolumn{2}{c|}{$\tau^{\textrm{imd}}$ = 1} & \multicolumn{2}{c|}{$\tau^{\textrm{cum}}$ = 1.5} & \multicolumn{2}{c|}{$\tau^{\textrm{flip}}$ = 0.5} & \\
\cmidrule(r){3-4} \cmidrule(r){5-6} \cmidrule{7-8}
& & $\text{ave}(\widehat{\tau})$&$\text{snr}(\widehat{\tau})$ & $\text{ave}(\widehat{\tau})$&$\text{snr}(\widehat{\tau})$ & $\text{ave}(\widehat{\tau})$&$\text{snr}(\widehat{\tau})$ & \\
\midrule
\textit{Standard Treatment imd}   & NO    &  0.969&\textbf{0.833}    &  \multicolumn{2}{c|}{N/A} &  \multicolumn{2}{c|}{N/A} & \textbf{1.536} \\
\midrule
\textit{Standard Treatment cum}   & NO    &  \multicolumn{2}{c|}{N/A} &  1.525&\textbf{0.751}    &  \multicolumn{2}{c|}{N/A} & \textbf{2.268}  \\
\midrule
\textit{Our Rapid Interleaving}      & YES   &  1.017&\textbf{1.032}    &  1.516&\textbf{0.770}    &  0.517&\textbf{2.390}    & \textbf{2.918}  \\
\bottomrule
\end{tabular}
}
\end{table}

We compare a few aspects of these three N-of-1 designs, summarized in Table~\ref{tab:A_vs_B_compare_designs} and Figure~\ref{fig:A_vs_B_compare_designs} below. First and the utmost, \textit{Our Rapid Interleaving} allows for comparing any counterfactual treatment paths combining Drugs A and B, that is, $\tau(q) - \tau(q'), q, q' \in \{ 0, 1\}^T$. The \textit{flip effect} serves as a running example; the two \textit{Standard Treatments} cannot be used to evaluate such \textit{flip effect}. As seen in the last row of Table~\ref{tab:A_vs_B_compare_designs} and Figure~\ref{fig:A_vs_B_snr}, our design can estimate all three estimands, whereas the other two standard designs are tailored for a specific estimand and cannot be used to estimate other estimands (shown as `N/A' cell in Table~\ref{tab:A_vs_B_compare_designs} and Figure~\ref{fig:A_vs_B_compare_designs}).

We run 5000 Monte Carlo simulations with $A = 2, B = 1$, and $g_t = 0.5^t$, each with an experiment period of $T = 35$ days, or 5 weeks. We numerically calculate the average and standard deviation for our convolutional estimator and the standard Horvitz-Thompson estimator. We compare their signal-to-noise ratios, defined as $\mathrm{SNR} = \mathrm{ave}(\widehat{\tau})/\mathrm{std}(\widehat{\tau})$, reported in Table~\ref{tab:A_vs_B_compare_designs} (as numerical values) and Figure~\ref{fig:A_vs_B_snr} (as histograms). Again, unlike the other designs, \textit{Our Rapid Interleaving} can estimate all three estimands with better accuracy and higher SNR.

In summary, the rapid interleaving design handles temporal interference, requires roughly similar time for testing which drug is better, and provides patients with better average outcomes during the experiment.

The above numerical example is supported by our variance formula~\eq{mom-variance}. Consider $\tau^{\textrm{imd}}$ as an example, which corresponds to $\tau(q)$ with $q_0=1$ and all other entries equalling 0. Recall the SNR
$$	
	\mathrm{SNR} = \E[\widehat{\tau}(q)]/\sqrt{\operatorname{var}(\widehat{\tau}(q))} \;.
$$
Suppose we observe an observation length $T$. Using the variance formula~\eq{mom-variance}, we have
\begin{align*}
	\operatorname{var}(\widehat{\tau}^{\textrm{imd}}) =  \frac{(A+B)^2}{T} \left\{ \left(\frac{A-B}{A+B}\right)^2 v(g) + w(g) \right\} \;, ~\text{where}~
	v(g) = \sum_{t=1}^{T-1} \tfrac{T-t}{T} g_t^2 \;, ~ w(g)= \tfrac{1}{T}||g * 1||_2^2 \;.
\end{align*}
Note that both quantities are bounded when the values of $g$ are summable. Thus, even with rapid interleaving of treatments, the SNR tends to infinity at a rate of $\sqrt{T}$.

For the conservative design where we wait $\ell$ days between treatments, the SNR with random treatment assignments is
$$
	\frac{A-B}{A+B} \sqrt{T} \ell^{-1/2}\,.
$$
We can compare the SNRs of the rapid design to the conservative design for $g_t = \beta^t$ with $\beta = 0.5$. As in the above example, let us choose $\ell = 5$ to wash out interference, setting $\ell^{-1/2}$ to be $0.44$. Hence, the SNR of the standard Horvitz-Thompson estimator is
$$
	0.44 \frac{A-B}{A+B} \sqrt{T} \,,
$$
where $T$ is the total number of days the experiment is run.

On the other hand, for the rapid interleaving design, using our convolutional estimator, the SNR for $\widehat{\tau}^{\textrm{imd}}$ would be at least
$$
	0.48 \frac{A-B}{A+B} \sqrt{T} \,.
$$
This bound is estimated by $v(g) \leq \tfrac{\beta^2}{1-\beta^2} \leq 0.34$ and $w(g) \leq \tfrac{1}{(1-\beta)^2} = 4$. This lower bound on the SNR is about 1.1 times that required for a standard design of the same duration, regardless of $A, B$. This shows that rapid interleaving of treatments may even accelerate the time required to assess the difference between treatments with decaying effects.

\subsection{Plug-in Variance Estimators}\label{sec:plugin-linear}

To construct confidence intervals using a normal approximation, we need to establish data-driven estimators for the variance of the estimator. Our approach is to provide direct, coarse estimates of the parameter vector $g$ and the error signal $e$, and then to plug these values into the variance formula~\eq{mom-variance}.

Within this section, we assume $q$ is only nonzero in the first $K$ coordinates. As is evident from the variance formula~\eqref{eq:mom-variance}, the variance of the estimator is potentially small only when the function $q$ is only nonzero in the first few components. We will describe precisely how small $K$ needs to be in more detail in Section~\ref{sec:circ_theory}.

For any integer $k$, let $u_k$ denote the $k$th standard basis vector. Then $\widehat{\tau}(u_k)$ is an unbiased estimator of $g_k$. Consider the following truncated estimator of $g$:
\begin{equation*}
	\widehat{g}_{< K} = (\widehat{\tau}(u_0),\widehat{\tau}(u_1),\ldots, \widehat{\tau}(u_{K-1}),0,\ldots,0) \,.
\end{equation*}
This estimator is unbiased in the first $K$ components but approximates the remaining components as equal to zero. We can use this estimate of $g$ to estimate the error signal $e$ as well
\begin{align*}
	\widehat{e}_t := \by_t - (\bx * \widehat{g}_{<K})_t \;,
\end{align*}

With these two estimates, we can estimate the variance. We simply plug these estimates into the variance formula~\eq{mom-variance}:
\begin{equation}\label{eq:mom-variance-plugin}
	\widehat{\cV} := 	\|\mathcal{T}_q^\transp \mathcal{T}_{\widehat{g}_{<K}}\|_F^2 
		+ \trace(\mathcal{T}_q^\transp \mathcal{T}_{\widehat{g}_{<2K}} \mathcal{T}_q^\transp \mathcal{T}_{\widehat{g}_{<2K}})
		- 2 \sum_{t \in [T]} (\mathcal{T}_q^\transp \mathcal{T}_{\widehat{g}_{<K}})_{tt}^2 
		+  \| \mathcal{T}_q^\transp \left( \widehat{g}_{<K} * \mathbf{1}+2 \widehat{e} \right)\|^2 \,.
\end{equation}
Note that in the second term, we use a $2K$-length estimate of $g$, but all other terms only use $K$ terms. The validity of this plug-in variance estimate will become transparent in the proof of Proposition~\ref{prop:plug-in-estimate}. In the subsequent theory section, we provide evidence for this estimator being asymptotically normal, and hence yielding valid normal confidence intervals.

\subsection{Connection to Existing Approaches}
\label{sec:comparison}

Finally, we close this section by discussing the existing approaches to dynamic causal inference and positioning our work in this growing literature.

\paragraph{Dynamic treatments}
To study observational longitudinal data \citep{robins_new_1986}, Robins and colleagues \citep{robins1999estimation, robins2000marginal, robins2004optimal} introduced structural nested models and designed fairly general identification and estimation methods for dynamic treatment effects. With a target expected response in mind, optimal sequential treatment rules have also been studied \citep{murphy2003optimal, robins2004optimal}. The framework to study dynamic treatments is generic; it considers various settings where the current treatment can interact with past treatments, outcomes, and covariates \citep{robins2000marginal}. However, in practice, some parsimonious models for the treatment-response relationship and models for how the current treatment interacts with the past are needed for meaningful identification and estimation, see \cite[Section 7, Equations 12-14]{robins2000marginal} and \cite{robins1999estimation}. 

Our work here considers a considerably narrower model for a specific, restricted class of dynamic treatments. In exchange for this narrowness, we are able to provide computable estimators for interesting estimands with reasonable variance in our settings. We consider a setup used in switchback experiments \citep{bojinov2022DesignAnalysis} and postulate a linear interference model for temporal effects drawing from linear dynamical systems. The temporal effects are also called impulse response functions \citep{granger1969investigating,sims1972money}. We study randomization inference of the impulse response functions.

\paragraph{Time series experiments}
\cite{bojinov2019time, bojinov2022DesignAnalysis} studied time series experiments and inference of temporal causal effects. The authors consider general causal effects indexed by any two treatment paths, and study estimation of the causal effects by inverse probability weighting according to each treatment path. Like the general interference model in Equation~\eqref{eqn:discrete-fourier}, \cite{bojinov2019time} do not pose any parsimonious structure on the response. Recall the fact that only 1 out of $2^T$ treatment paths is observed; In general, there is no hope of estimating the effect of an arbitrary counterfactual path. Therefore, for a fixed time $k$, the authors make the estimand depend on the observed path up to time $T-k$ for valid estimation \cite[Sections 3.2 and 3.3]{bojinov2019time}. This eliminates a large number of counterfactual paths from the analysis. Building on top of their prior work, \cite{bojinov2022DesignAnalysis} explored the minimax experimental design of regular switchback experiments. 

The experimental estimands in \cite{bojinov2019time} and \cite{bojinov2022DesignAnalysis} focus on average short-term effects, whereas our estimand considers any linear functional of the long-term effects. Our estimation method can evaluate any counterfactual path because we impose a linear interference model, thus introducing a parsimonious structure on the response. The salient distinction from the present work is the focus on general response, short-term effects, versus linear response, long-term effects. \cite{hu2022switchback} also emphasized settings where temporal effects do not entirely vanish after a short term. As for confidence intervals for inference, \cite{bojinov2019time} only considered the sharp null where there are no temporal effects, whereas our inference result holds for any temporal effects with sufficient decay. In summary, our papers complement each other and point to the potential of a general inference framework in N-of-1 trials.

\paragraph{Markov decision processes}
Markov decision processs (MDP) is another choice to model the temporal interference in dynamic causal inference. Assuming full or partial knowledge of the state of the system, one can estimate the MDP \citep{glynn2020adaptive} from data, and further estimate treatment effects, see \cite{liao2021off, liao2022batch}, \cite{kallus2020double, kallus2022efficiently}, and \cite{hu2023off}. 

The model we consider in this work is a \emph{partially observed} Markov decision process with continuous actions, states, and outputs. This makes our work considerably different in character than prior work on MDPs which have observed, discrete actions, states, and outputs. We are able to move to the partially observed setting because we restrict our attention to linear models. In these models, we make no assumption about the dimension of the hidden state, which could be infinite. Our approach complements the MDP approach, taking advantage of the special properties of linear dynamical systems when modeling temporal interference.

\paragraph{Panel data methods}
Though not the subject of this paper, causal inference with panel data \citep{arkhangelsky2023large, bojinov2021panel} has received broad interest in applied econometrics, where emphases are made on heterogeneity across units. We refer the readers to the extensive survey by \cite{arkhangelsky2023causal}. There are two significant differences between N-of-1 studies and panel methods. First, in panel data methods, the assignment or treatment mechanisms are typically structured, for instance, block assignment, staggered adoption, and clustered assignment \cite[Section 3.3]{arkhangelsky2023causal}. In contrast, N-of-1 trials vary treatments over time and focus on temporal interference. Second, panel data methods typically model the heterogeneity and form estimands in an average sense. In contrast, we run time series experiments for each unit and study randomization inference of individual impulse responses.

\section{Circular Convolution Model and Asymptotic Normality} 
\label{sec:circ_theory}

In this section, we first develop the circular convolution model, providing the analogous estimand and estimator to the linear case. We show that for large $T$, the estimators for the linear and circular models differ by a vanishing amount. We then establish asymptotic normality for the circular model, which in turn verifies the asymptotic normality of the linear design as well. As we shall see, the variance formula for the circular case takes an interpretable form compared to the linear case.

Consider a circular convolution model for the impulse response with horizon size $T \in \mathbb{N}$, 
\begin{align}
	\label{eqn:circular-conv-model}
	\by = \bx \ast g + e \in \mathbb{R}^T \;. 
\end{align}
How does this differ from linear convolution? As we described in Section~\ref{sec:convolution}, circular convolution models $g$ and $e$ as periodic signals, and such models do not describe the sorts of N-of-1 experiments typically considered in the medical context. However, on a long time horizon, when $g$ decays quickly, circular convolutions and linear convolutions produce similar outputs. As we will see, shortly, circular convolutions will provide insights into the asymptotic behavior of linear convolution experiment designs, provided that the observation time is long enough.

Moreover, circular convolutions are interesting in their own right for modeling periodic phenomena. For instance, periodic behavior due to seasonality and other cyclic trends in financial times series (see, for example, \cite{box1978analysis} and Chapter 2.8 of \cite{tsay2005analysis}). Hence, we include full details for experiment design in such periodic systems.

\subsection{Estimand and Estimator}
We are concerned with a family of estimands indexed by $q \in \R^T$ 
\begin{align}
	\label{eqn:estimand}
	\tau(q) := \langle q, g \rangle \;.
\end{align}
Special cases include the \textit{cumulative lag-$K$ effects} with an integer $K$, where vector $q = \mathbf{1}_{<K} := \sum_{k=0}^{K-1} u_k$ is plugged in, 
\begin{align}
	\label{eqn:lag-K}
	\tau_{K} := \tau(\mathbf{1}_{<K}) = \tau(\sum_{k=0}^{K-1} u_k) = \sum_{k=0}^{K-1} g_k \;.
\end{align}
Again, here $u_k$ is the $k$-th standard basis vector.

Accordingly, consider estimators of the following form, indexed by $q \in \R^T$
\begin{align}
	\label{eqn:estimator} 
	\widehat{\tau}(q)= \tfrac{1}{T}\langle (2\bx-\mathbf{1}) \ast q , 2 \by \rangle \;.
\end{align}
As a special case, the estimator for the lag-$K$ effects $\tau_{K}$ is
\begin{align}
	\label{eqn:estimator-lag-K}
	\widehat{\tau}_K := \widehat{\tau}(\mathbf{1}_{<K}) = \tfrac{1}{T}\langle (2\bx-\mathbf{1}) \ast \mathbf{1}_{<K} , 2 \by \rangle \;.
\end{align}

For the linear convolution model, the estimator takes the same form as the \eqref{eqn:estimator}, only changing the circular convolution to linear convolution,
\begin{align}
	\label{eqn:estimator-lin}
	\widehat{\tau}^{\mathcal{L}}(q)= \tfrac{1}{T}\langle (2\bx-\mathbf{1}) * q , 2 \by \rangle \;.
\end{align}
In this case, the estimand turns out to be slightly different than that in \eqref{eqn:estimand},
\begin{align}
	\label{eqn:estimand-lin}
	\tau^{\mathcal{L}}(q) :=  \sum_{t \in [T]} \tfrac{T-t}{T} q_t g_t \;.
\end{align}
However, one can already see that if $q_t$ is $0$ for $t\geq K$, then the linear and circular estimands approach each other as $T$ tends to infinity.

Let us now consider the same simple design as before, when $\bx \in \mathbb{R}^T$ is randomly treated at each time step. That is, $\bx_t, 0\leq t <T$ is i.i.d. Bernoulli. As before, we denote the centered randomization vector $\bz := 2\bx-\mathbf{1} \in \mathbb{R}^T$. The estimator $\widehat{\tau}(q)$ is a quadratic polynomial of the $\bz$:
\begin{align*}
	\widehat{\tau}(q) &= \tfrac{1}{T}\langle (2\bx-\mathbf{1}) \ast q , (2\bx-\mathbf{1}) \ast g \rangle + \tfrac{1}{T}\langle (2\bx-\mathbf{1}) \ast q , \mathbf{1} \ast g + 2e \rangle \;, \\
	&= \tfrac{1}{T}\langle \bz \ast q , \bz \ast g \rangle + \tfrac{1}{T}\langle \bz \ast q , h \rangle  \;,
\end{align*}
where $$h := \mathbf{1} \ast g + 2e\in \mathbb{R}^T\;.$$
Observe that the difference between the estimator and the estimand, noted as $\bW_t$, takes the form
\begin{align}
	\label{eqn:decomp}
	\bW_T:= \widehat{\tau}(q) - \tau(q) = \tfrac{1}{T} \sum_{i \neq j \in [T]} \bz_i \bz_j H_{ij} + \tfrac{1}{T} \sum_{i  \in [T]} \bz_i L_i \;,
\end{align}
where
\begin{align}
	\label{eqn:H-and-L}
	H_{ij} := \sum_{t \in [T]} q^\circ_{t-i} g^\circ_{t-j}, ~L_i := \sum_{t \in [T]} q^\circ_{t-i} h_{t} \;.
\end{align}
Here the $^\circ$ notation denotes the circular extension defined in \eqref{eqn:circ-abbrev}.
This calculation reveals that the estimator is unbiased.  The notation also allows us to connect back to the linear model.

\subsection{Circular vs. Linear Convolutions}
\label{sec:circular_vs_linear}
This section shows that circular and linear convolutions are intimately connected, at least for the type of impulse response we consider.
Consider an impulse response function $g \in \mathbb{R}^T$.  As developed in Section~\ref{sec:linear-conv}, the estimator for the linear convolution model was given by
\begin{align*}
	\bW_T^{\mathcal{L}}:= \widehat{\tau}^{\mathcal{L}}(q) - \tau^{\mathcal{L}}(q)
	= \tfrac{1}{T} \sum_{i \neq j \in [T]} \bz_i \bz_j H_{ij}^{\mathcal{L}} + \tfrac{1}{T} \sum_{i  \in [T]} \bz_i L_i^{\mathcal{L}} \;,
\end{align*}
where
\begin{align*}
	H^\mathcal{L} &:=  \mathcal{T}_q^\transp \mathcal{T}_g \;, \\
	L^{\mathcal{L}} &:= \mathcal{T}_q^\transp h^{\mathcal{L}}, ~\text{where}~h^{\mathcal{L}} := \mathbf{1} * g+ 2 e \;.
\end{align*}
We can already see that the linear term $L^\mathcal{L}$ closely imitates the circular convolution model and exhibits asymptotic normality under mild conditions.

We now show that the matrix $H^\mathcal{L}$ and $H$ are entrywise close.
\begin{lemma}
Recall $H_{ij}$'s defined in \eqref{eqn:H-and-L}, and assume that $q, g$ are supported only on the first $K$-entries with each entry bounded. Then
\begin{align*}
	\sup_{z \in \{-1,1\}^T}  \tfrac{1}{T}\left| z^\transp ( H - H^{\mathcal{L}}) z \right|= O \left(\tfrac{K^3}{T} \right) \;.
\end{align*}
\end{lemma}
\begin{proof}
Note that we can write the entries of both matrices in a parallel manner
\begin{align*}
	H_{ij} = \langle \mathcal{C}_q u_i, \mathcal{C}_g u_j \rangle \;, \; H_{ij}^{\mathcal{L}} = \langle \mathcal{T}_q u_i, \mathcal{T}_g u_j \rangle\,.
\end{align*}
These are dot products between columns of circulant and Toeplitz matrices, respectively. Let $v$ be a vector supported only on the first $K$ entries. Then $\mathcal{C}_v u_i = \mathcal{T}_v u_i$ whenever $i < T-K$. Hence, $H_{ij}=H_{ij}^{\mathcal{L}}$ unless either $i$ or $j$ is greater equal than $T-K$. We also have that both $H_{ij}$ and $H^{\mathcal{L}}_{ij}$ equal $0$ if $|i-j  \pmod{T} | >K$ due to the fact that $q, g$ are supported only on the first $K$-entries. By the same reasoning, $\max_{i,j} |H_{ij}|$ and $\max_{i,j} |H^{\mathcal{L}}_{ij}|$ are both $O(K)$. Hence $H_{ij}-H_{ij}^{\mathcal{L}}$ is only nonzero in at most $O(K^2)$ entries, and in these entries, the magnitude of the difference is at most $O(K)$. This proves the lemma. 
\end{proof}

In what follows, we shall establish $\tfrac{1}{\sqrt{T}}$ asymptotic normality for the circular convolution model where the higher order moment calculations and formulae are considerably simpler. Coupled with the above approximation result, asymptotic normality holds for the linear model.

\subsection{Variance of the Circular Convolution Model}
\label{sec:higher-moments}

In the circular convolution model, the variance of $\widehat{\tau}(q)$ yields a more interpretable expression than the linear model. Since the two estimators are close to each other for large $T$, we can draw further insights about how the variance depends on $g$ and $q$.

\begin{proposition}[Second Moment]
	\label{prop:second-moment}
\begin{align*}
	T \cdot \E_{\bx} \left[ \big( \widehat{\tau}(q) -  \tau(q) \big)^2 \right] = \cV_{Q} + \cV_{L} \;,
\end{align*}
where
\begin{align}\label{eq:delta-var-decomposition}
	\cV_{Q}&= \operatorname{var}\left(\frac{1}{\sqrt{T}} \sum_{i \neq j \in [T]} \bz_i \bz_j H_{ij}\right)= \| g \ast q \|_2^2 + \langle g \ast g, q \ast q \rangle - 2 \langle q, g \rangle^2  \\
	 \cV_{L}&=  \operatorname{var}\left( \tfrac{1}{\sqrt{T}} \sum_{i  \in [T]} \bz_i L_i\right)=  \tfrac{1}{T}  \| (\mathbf{1} \ast g + 2e) \ast  q \|_2^2\,.
\end{align}
\end{proposition}

The variance of this estimator thus is governed by the norm of the convolution of $g$ with $q$ and on the norm of the exogenous disturbance $e$. In order for this variance to be small, it is necessary that $g$ and $q$ have norms that grow sublinearly with $T$. In particular, this suggests that $q$ should have sublinear support size. The estimable treatment effects must only concern a small number of components of the impulse response $g$. Under these assumptions, we can further show that the estimator is asymptotically normal. Asymptotic normality will serve as a building block for inference of the causal parameters.

\subsection{Asymptotic Normality}
\label{sec:asmptotic-normality}

We establish asymptotic normality individually for the linear and quadratic terms of~\eqref{eqn:decomp}. To establish the asymptotic normality for the quadratic term, we need one technical assumption on the decay of signal $g$ and the vector $q$. It is crucial to note that asymptotic normality does not hold for arbitrary $g$ and $q$; it holds when the signal $g\ast q$ decays sufficiently fast.
\begin{assumption}
	\label{asmp:main}
	Assume $g$ and $q$ satisfy
 $$\| (|g| \ast |q|) \ast (|g| \ast |q|) \|_2^2 \leq C \cdot T^{1-\epsilon} \cdot \big( \| g \ast q \|_2^2 + \langle g \ast g, q \ast q \rangle - 2 \langle q, g \rangle^2  \big)^2 $$
		with some universal constants $C>1$, $\epsilon>0$.
\end{assumption}

This assumption can be satisfied by various means; a sufficient condition is when the signal $g \ast q$ is approximately supported on the top-K elements with $K = K(T) \ll T$. One that most closely connects with assumptions in the dynamical systems literature is to assume $g$ is asymptotically stable. For the purpose of this work, this just means that $g_t$ can be upper bounded by an exponentially decaying function. Indeed, suppose $|g_t| \precsim \rho^{t}$ with some positive $\rho \in (0, 1)$ and $\| g \ast q \|_2^2 + \langle g \ast g, q \ast q \rangle - 2 \langle q, g \rangle^2  = \Omega(1)$, a non-trivial limiting variance, then the above assumption holds with a wide range of $K = K(T) = \mathrm{Polylog}(T)$. This is true because $\| (|g| \ast |q|) \ast (|g| \ast |q|) \|_2^2 \precsim \mathrm{Poly}(K) = \mathrm{Polylog}(T) \precsim  T^{1-\epsilon}$. The constant $\rho$ is related to the mixing time of the dynamical system when driven by white noise~\cite{simchowitz2018learning}.

With Assumption~\ref{asmp:main} established, we can now proceed to verify asymptotic normality. Denote
\begin{align*}
		\bH_{T} :&=  \tfrac{1}{T}\sum_{i \neq j \in [T]} \bz_i \bz_j H_{ij} \;,
		\bL_{T} := \tfrac{1}{T} \sum_{i  \in [T]} \bz_i L_i\,,
\end{align*}
where $H_{ij}$ and $L_i$ are defined in \eqref{eqn:H-and-L}.

The following asymptotic normality holds for generic $g$ and $q$ satisfying Assumption~\ref{asmp:main}.
\begin{theorem}[Quadratic Term]
	\label{thm:Quadratic-Term}
	Under Assumption~\ref{asmp:main},
	\begin{align*}
		\frac{\sqrt{T} \cdot \bH_{T}}{ \sqrt{ \| g \ast q \|_2^2 + \langle g \ast g, q \ast q \rangle - 2 \langle q, g \rangle^2 } } \Rightarrow \cN(0, 1), \quad \text{as $T\rightarrow \infty$} \;.
	\end{align*}
\end{theorem}

In order to prove this theorem, we need to record two lemmas. These are both proven in the appendix. The first lemma controls the magnitude of terms in the second moment of $\mathbf{H}$. This is needed in multiple places in order to verify asymptotic normality.
\begin{lemma}\label{lemma:H-inf-norm}
With $H_{ij}$ defined as in~\eqref{eqn:H-and-L},
\begin{align*}
		\max_{i} \sum_{j: j\neq i} H_{ij}^2+H_{ij}H_{ji} \leq \| g \ast q \|_2^2 + \langle g \ast g, q \ast q \rangle -2 \langle q, g \rangle^2\,.
\end{align*}
\end{lemma}

The second lemma provides the critical moment condition bounding the ratio of the fourth moment and square of the second moment of $\bH$.
\begin{lemma}[Fourth Moment Estimate]
	\label{lemma:fourth-moment}
Then
	\begin{align*}
		\left| \frac{\E[\bH_{T}^4 ]}{\big( \E [\bH_T^2] \big)^2 } -   3 \right| \leq \frac{4}{T} + \frac{16}{T} \frac{\| (|g| \ast |q|) \ast (|g| \ast |q|) \|_2^2}{\left(\| g \ast q \|_2^2 + \langle g \ast g, q \ast q \rangle - 2 \langle q, g \rangle^2 \right)^2} \;.
	\end{align*}
\end{lemma}

Together with Assumption~\ref{asmp:main}, these two lemmas immediately yield a proof of Theorem~\ref{thm:Quadratic-Term} via the argument of DeJong~ \citep{dejong1987CentralLimit}.
\begin{proof}[Proof of Theorem~\ref{thm:Quadratic-Term}]
    We first verify the assumption (a) in Theorem 2.1 of \cite{dejong1987CentralLimit}. By Lemma~\ref{lemma:H-inf-norm},
	\begin{align*}
		\frac{\max_{i} \sum_{j: j\neq i} H_{ij}^2+H_{ij}H_{ji}}{\sum_{i\neq j} H_{ij}^2+H_{ij}H_{ji}} \leq \frac{1}{T}  \,,
	\end{align*}
	which goes to zero as $T\rightarrow \infty$.

By Lemma~\ref{lemma:fourth-moment}
	\begin{align*}
		\left| \frac{\E[\bH_{T}^4 ]}{\big( \E [\bH_T^2] \big)^2 } -   3 \right| \leq \frac{4}{T} + \frac{16}{T} \frac{\| (|g| \ast |q|) \ast (|g| \ast |q|) \|_2^2}{\left(\| g \ast q \|_2^2 + \langle g \ast g, q \ast q \rangle - 2 \langle q, g \rangle^2 \right)^2} \;.
	\end{align*}
	By Assumption~\ref{asmp:main}, we have
	\begin{align*}
		\frac{\| (|g| \ast |q|) \ast (|g| \ast |q|) \|_2^2}{\left(\| g \ast q \|_2^2 + \langle g \ast g, q \ast q \rangle - 2 \langle q, g \rangle^2 \right)^2} \leq C \cdot T^{1-\epsilon}
	\end{align*}
	and thus
	\begin{align*}
		\lim_{T \rightarrow \infty} \frac{\E[\bH_{T}^4 ]}{\big( \E [\bH_T^2] \big)^2 } = 3 \;.
	\end{align*}
	Now, we have verified the assumption (b) in Theorem 2.1 of \cite{dejong1987CentralLimit}. Therefore, asymptotic normality is a direct consequence of Theorem 2.1 in \cite{dejong1987CentralLimit}.
\end{proof}

Asymptotic normality of the linear term follows from the standard Lindeberg-Feller Theorem.
\begin{theorem}[Linear Term]
	\label{lem:Linear-Term}
	Assume that $g, e, q \in \mathbb{R}^T$ satisfy $|\langle \mathbf{1}, g  \rangle \langle \mathbf{1}, q \rangle| >2 \| e \|_{\infty} \| q\|_1$.
Then
	\begin{align*}
		\frac{\sqrt{T} \cdot \bL_{T}}{\sqrt{\tfrac{1}{T}\| (\mathbf{1} \ast g + 2e) \ast  q \|_2^2}} \Rightarrow \cN(0, 1), \quad \text{as $T\rightarrow \infty$} \;.
	\end{align*}
\end{theorem}

\begin{proof}[Proof of Theorem~\ref{lem:Linear-Term}]
We can write
	\begin{align*}
		\frac{\sqrt{T}\cdot \bL_{T}}{\sqrt{\tfrac{1}{T}\| (\mathbf{1} \ast g + 2e) \ast  q \|_2^2}} = \sum_{i\in [T]} \bz_i \frac{ \tfrac{1}{\sqrt{T}} L_i}{ \sigma_T} \,,
	\end{align*}
where
	\begin{align*}
		\sigma_T^2 &:= \tfrac{1}{T}\sum_{i\in [T]} L_i^2 = \tfrac{1}{T}  \| (\mathbf{1} \ast g + 2e) \ast  q \|_2^2 \\
		&= \tfrac{1}{T}\ \|  \langle \mathbf{1}, g \rangle  \langle \mathbf{1}, q \rangle  \mathbf{1} + 2e \ast q \|_2^2
	\end{align*}
We know
	\begin{align*}
		\max_i |L_i| &\leq \| \langle \mathbf{1}, g \rangle   \langle \mathbf{1}, q \rangle  \mathbf{1} + 2e \ast q \|_{\infty} \leq   |\langle \mathbf{1}, g \rangle \langle \mathbf{1}, q \rangle | + 2\| e \ast q \|_{\infty} \\
		\sigma_T^2 &=  \langle \mathbf{1}, g \rangle^2  \langle \mathbf{1}, q \rangle^2 + 4 \langle \mathbf{1}, g \rangle  \langle \mathbf{1}, q \rangle   \tfrac{1}{T} \langle \mathbf{1}, e \ast q \rangle + \tfrac{4}{T} \langle e \ast q, e \ast q \rangle
	\end{align*}
	Under the assumption, there exists a small constant $\epsilon$ such that (1) $ | \langle \mathbf{1}, g \rangle   \langle \mathbf{1}, q \rangle | > \tfrac{1}{(1-\epsilon)^2}  2\| e \ast  q \|_{\infty}$ and (2)  $ \langle \mathbf{1}, g \rangle^2   \langle \mathbf{1}, q \rangle^2   >  \tfrac{1}{(1-\epsilon)^2} \tfrac{4}{T} \langle e \ast q, e \ast q \rangle$, and thus
\begin{align*}
		\max_i |L_i|  \leq  (1+ (1-\epsilon)^2) | \langle \mathbf{1}, g \rangle   \langle \mathbf{1}, q \rangle  |  \,.
\end{align*}

Now, using the inequality
\begin{align*}
	&(1-\epsilon)  \langle \mathbf{1}, g \rangle^2  \langle \mathbf{1}, q \rangle^2  +  \tfrac{4}{1-\epsilon} \tfrac{1}{T}\langle e \ast q, e \ast q \rangle + 4 \langle \mathbf{1}, g \rangle   \langle \mathbf{1}, q \rangle   \tfrac{1}{T} \langle 1, e \ast q \rangle \\
	&=\tfrac{1}{T} \| \sqrt{1-\epsilon}   \langle \mathbf{1}, g \rangle   \langle \mathbf{1}, q \rangle   \mathbf{1} + 2 \tfrac{1}{\sqrt{1-\epsilon}} e \ast q \|^2    \geq 0\;,
\end{align*}
we can lower bound the variance
\begin{align*}
		\sigma_T^2 &\geq   \langle \mathbf{1}, g \rangle^2  \langle \mathbf{1}, q \rangle^2  +  \tfrac{4}{T} \langle e \ast q, e \ast q \rangle - (1-\epsilon)   \langle \mathbf{1}, g \rangle^2  \langle \mathbf{1}, q \rangle^2   - \tfrac{4}{1-\epsilon} \tfrac{1}{T}\langle e \ast q, e \ast q \rangle \\
		&\geq \epsilon  \langle \mathbf{1}, g \rangle^2  \langle \mathbf{1}, q \rangle^2 - \tfrac{\epsilon}{1-\epsilon}\tfrac{4}{T} \langle e \ast q, e \ast q \rangle \geq \epsilon^2  \langle \mathbf{1}, g \rangle^2  \langle \mathbf{1}, q \rangle^2 \,.
	\end{align*}

We hence have
\begin{align*}
		\max_i \frac{\tfrac{1}{T}|L_i|^2}{ \sigma_T^2} \leq  \tfrac{1}{T} \frac{(1+ (1-\epsilon)^2)^2  \langle \mathbf{1}, g \rangle^2   \langle \mathbf{1}, q \rangle^2 }{\epsilon^2   \langle \mathbf{1}, g \rangle^2   \langle \mathbf{1}, q \rangle^2 } \rightarrow 0\;,
\end{align*}
	which implies for any fixed $\delta$
	\begin{align*}
		\limsup_{T \rightarrow \infty} \frac{1}{\sigma_T^2} \sum_{i}  \tfrac{1}{T}|L_i|^2 \mathbb{P}( |\bz_i | \tfrac{1}{\sqrt{T}} L_i > \delta \sigma_T ) = 0 \;.
	\end{align*}
	The last equation verifies the Lindeberg-Feller condition; thus, asymptotic normality holds.
\end{proof}

\section{Inference: Confidence Intervals}

Equipped with the asymptotic normality, building confidence intervals is reduced to estimating the variance. In this section, we will extend the plug-in variance estimator of Section~\ref{sec:plugin-linear} to the circular case and show that it leads to valid confidence intervals. The argument for the linear convolution model is exactly the same, albeit with less compact formulas in the calculations.

We will use non-asymptotic concentration inequalities to derive that both terms $\cV_{Q}$ and $\cV_{L}$ in Proposition~\ref{prop:second-moment} can be estimated consistently with high probability. 
With the asymptotic normality established in the previous section, we thus can derive data-driven confidence intervals on $\widehat{\tau}(q)$. To streamline the presentation within this section, we focus on the case when $q$ is supported on the top-K entries, a class of estimands including cumulative lag-K effects.
We will provide estimates
$\widehat{\cV}_{Q}$ and $\widehat{\cV}_{L}$ for the quantities $\cV_{Q}$ and $\cV_{L}$ in the variance formula of Propositon~\ref{prop:second-moment}. Then provided $K=O(\log(T))$,
\begin{align*}
	\widehat{\tau}(q) \pm \mathcal{Z}_{1-\alpha/2} \sqrt{\tfrac{\widehat{\cV}_{Q} + \widehat{\cV}_{L}}{T}} \;.
\end{align*}
will be valid confidence intervals with asymptotic coverage $1-\alpha$.

To start with, we state a uniform consistency result.
\begin{theorem}[Uniform Consistency]
	\label{thm:uniform-consistency}
	Assume $|e_t| \leq M,~\forall t \geq 0$ and $ \| g \|_1 \leq C$ for absolute constants $C, M>0$.  Let $K$ be an integer that $K\ll T$, and recall $u_k, k\in [K]$ defined in \eqref{eqn:lag-K}.
	Then, with probability at least $1 - 4 T^{-2}$
	\begin{align*}
		\sup_{k < K} \left| \widehat{\tau}(u_k) - g_k \right| &\precsim \frac{\log(T)}{\sqrt{T}} \;,\\
		\left| \widehat{\tau}(\mathbf{1}_{<K}) - \langle \mathbf{1}_{<K} , g \rangle \right|  &\precsim \frac{K\log(T)}{\sqrt{T}} \;.
	\end{align*}
	Assume in addition that $|g_t| \precsim c^{t}$ with some positive constant $c<1$. Then with $K = K(T) = \Theta(\log(T))$,
	\begin{align*}
		\left| \widehat{\tau}(\mathbf{1}_{<K}) - \langle g, \mathbf{1} \rangle \right| \stackrel{a.s.}{\rightarrow} 0,  \quad \text{as $T\rightarrow \infty$} \;.
	\end{align*}
\end{theorem}

This Theorem can be proven by invoking the Hanson-Wright and Bernstein inequalities. Before proving the theorem, we record a useful inequality that explicitly captures the error rate of our estimator.
\begin{lemma}\label{lem:expoential-tail}
	With probability at least $1 - 4\exp(-\gamma)$, for any fixed $q$,
	\begin{align*}
		\left| \widehat{\tau}(q) - \langle g, q \rangle \right| &\precsim \sqrt{\frac{\| g \ast q \|_2^2 + \langle g \ast g, q \ast q \rangle}{T}} ( \sqrt{\gamma} \vee \gamma) \\
		& \quad + \sqrt{\frac{\| (\mathbf{1} \ast g + 2e) \ast q\|_2^2}{T^2}} \sqrt{\gamma} + \frac{\|q\|_1 \cdot \| \mathbf{1} \ast g + 2e \|_{\infty}}{T} \gamma \;.
	\end{align*}
\end{lemma}

\begin{proof}
Using~\eqref{eqn:decomp}, we have
\begin{align*}
\left| \widehat{\tau}(q) - \langle g, q \rangle \right| \leq  \tfrac{1}{T} \left|\sum_{i \neq j \in [T]} \bz_i \bz_j H_{ij}\right| + \tfrac{1}{T} \left|\sum_{i  \in [T]} \bz_i L_i \right|\;.
\end{align*}
We can bound each term separately. By the Hanson-Wright inequality, we have with probability at least $1 - 2\exp(-\gamma)$
\begin{align*}
		\left|\sum_{i \neq j} H_{ij} \bz_i \bz_j - \mathrm{Tr}(H)\right| \precsim \sqrt{\| H \|_{\mathrm{F}}^2 \gamma} + \| H \|_{\mathrm{op}} \gamma \;,
\end{align*}
As derived in the proof of Proposition~\ref{prop:second-moment},
\begin{align*}
	\left\| \tfrac{1}{2} (H + H^\transp) \right\|_{\mathrm{F}}^2 = \tfrac{T}{2} \big(\| g \ast q \|_2^2 + \langle g \ast g, q \ast q \rangle - 2\langle g, q \rangle^2 \big) \,.
\end{align*}
and we always also have $\| H + H^\transp  \|_{\mathrm{op}} \leq \| H + H^\transp  \|_{\mathrm{F}}$.

For the second term, with probability at least $1 - 2\exp(-\gamma)$, Bernstein's inequality asserts
\begin{align*}
		|\sum_{i} L_i \bz_i| \precsim \sqrt{\| L \|_2^2 \gamma} + \| L \|_{\infty} \gamma \;.
\end{align*}
In this case, we again know from Proposition~\ref{prop:second-moment} that $\| L \|_2^2 = \| (\mathbf{1} \ast g + 2e) \ast q\|_2^2$ and again bound $\| L \|_{\infty} \leq \| q \|_{1} \cdot \| \mathbf{1} \ast g + 2e \|_{\infty}$, completing the proof.
\end{proof}

Using this lemma, Theorem~\ref{thm:uniform-consistency} is immediate

\begin{proof}[Proof of Theorem~\ref{thm:uniform-consistency}]
	For $q = u_k$, one can immediately verify that
	\begin{align*}
		\| g \ast q \|_2^2 &= \sum_{i} (g_{i+k})^2 = \| g \|_2^2 \;, \\
		\langle g \ast g, q \ast q \rangle &\leq \sum_i g_{i} g_{2k-i} \leq \big(\sum_i g_{i}^2 \big)^{1/2}\big(\sum_i (g_{2k-i})^2 \big)^{1/2} =  \| g \|_2^2 \;.
	\end{align*}
	In addition,
	\begin{align*}
		\frac{\| (\mathbf{1} \ast g + 2e) \ast q\|_2^2}{T} = \frac{ \|  \mathbf{1} \ast g + 2e \|_2^2}{T} \leq  \| \mathbf{1} \ast g + 2e \|_{\infty}^2 \;.
	\end{align*}
	Lemma~\ref{lem:expoential-tail} reads:
	\begin{align*}
		 \left| \widehat{\tau}(u_k) - g_k \right| \precsim \|g \|_2 \frac{\sqrt{\gamma} \vee \gamma}{\sqrt{T}} + \| \mathbf{1} \ast g + 2e \|_{\infty} ( \sqrt{\frac{\gamma}{T}} \vee \frac{\gamma}{T} ) \;,
	\end{align*}
	with probability at least $1-4\exp(-\gamma)$. We complete the proof by applying union bound and plug in $\gamma = \log(T^2 K)$. To obtain the almost sure statement $\widehat{\tau}(\mathbf{1}_{<K}) - \langle g, \mathbf{1} \rangle \stackrel{a.s.}{\rightarrow} 0$, we take $K = \Theta(\log(T))$ and apply the Borel-Cantelli lemma since $\sum_{T} T^{-2} < \infty$.

\end{proof}

\paragraph{Plug-in estimate of variance}
Now we are ready to state the data-driven estimators for the terms $\cV_{Q}$ and $\cV_{L}$ in the case of convolution. 
The estimators are essentially the same as what was presented for linear convolution. Denote
\begin{align*}
	\widehat{g}_{<K} = (\overbrace{\widehat{\tau}(u_0), \ldots, \widehat{\tau}(u_{K-1})}^{\text{first $K$ entries}}, 0 \ldots, 0 )
\end{align*}
and define the plug-in estimate of $\cV_Q$
\begin{align}
	\label{eqn:est-VQ}
	\widehat{\cV}_{Q}  &= \| \widehat{g}_{<K} \ast q \|_2^2 + \langle  \widehat{g}_{<2K} \ast  \widehat{g}_{<2K}, q \ast q \rangle - 2 \langle  q, \widehat{g}_{<K} \rangle^2
\end{align}
For the term $\cV_L$, we are going to design a plug-in estimate. We introduce the estimator for the error vector
	\begin{align*}
		\widehat{e}_t := \by_t - (\bx \ast \widehat{g}_{<K})_t \;,
	\end{align*}
and define
\begin{align}
	\label{eqn:est-VL}
	\widehat{\cV}_{L}  &= \tfrac{1}{T}  \| (\mathbf{1} \ast \widehat{g}_{<K} + 2 \widehat{e}) \ast  q \|_2^2 \;. 
\end{align}
Conceptually, the plug-in estimate of $\cV_{Q}, \cV_{L}$ is to replace $g \leftarrow \widehat{g}_{<K}$, and to replace $e \leftarrow \widehat{e}:= \by - \bx \ast \widehat{g}_{<K}$.

\begin{proposition}[Validity of the Plug-in Estimate of Variance]
	\label{prop:plug-in-estimate}
	Assume $q$ is supported on the top-K entries with $\| q \|_{\infty} \leq 1$, and that $|g_t| \precsim c^{t}$ with some positive constant $c<1$. Then with $K = K(T) = \Theta(\log(T))$, as $T\rightarrow \infty$
	\begin{align*}
		\widehat{\cV}_{Q}  \stackrel{a.s.}{\rightarrow} \cV_{Q} \;,\\
		\widehat{\cV}_{L}  \stackrel{a.s.}{\rightarrow} \cV_{L} \;.
	\end{align*}
\end{proposition}
With the asymptotic normality established in the previous section and the above proposition asserting the validity of the plug-in estimate of variance, we can define the data-driven confidence interval with coverage $1-\alpha$:
\begin{align*}
	\widehat{\tau}(q) \pm C_{\alpha} \sqrt{\tfrac{\widehat{\cV}_{Q} + \widehat{\cV}_{L}}{T}} \;.
\end{align*}

\section{Empirics}

In this section, we demonstrate a simple empirical example illustrating two theoretical results established, the
asymptotic normality in Section~\ref{sec:asmptotic-normality}, and the asymptotic equivalence between circular and linear convolution in Section~\ref{sec:circular_vs_linear}. Besides validating the theory, the empirical example also demonstrates how asymptotic characterizations become accurate as $T$ grows.

We fix an impulse response $g_t = 1.00*0.65^t -1.60*0.50^t+0.75*0.48^t,~ t\in [T]$ and set $K=25$ in constructing the estimator $\widehat{\tau}_K$ defined in \eqref{eqn:estimator-lag-K}. For this impulse response, it is clear that the signals are negligible after $K=25$. We consider two data-generating processes (DGP), the circular convolution model in \eqref{eqn:circular-conv-model} (see Fig.~\ref{fig:circular-convolution}), and the linear convolution model in \eqref{eqn:linear-conv-model} (see Fig.~\ref{fig:linear-convolution}), both with $e = 0$. For each DGP, we sample the impulse/treatment vector $\bx \in \{0,1\}^T$ with i.i.d. Bernoulli with $T=200, 1000, 5000$ and then observe the response $\by \in \mathbb{R}^T$, to illustrate the asymptotic behavior. For each DGP and $T$, we run 20000 Monte Carlo simulations to show the histogram of $\widehat{\tau}_K$. 

\begin{figure}
     \centering
     \begin{subfigure}[b]{0.32\textwidth}
         \centering
         \includegraphics[width=\textwidth]{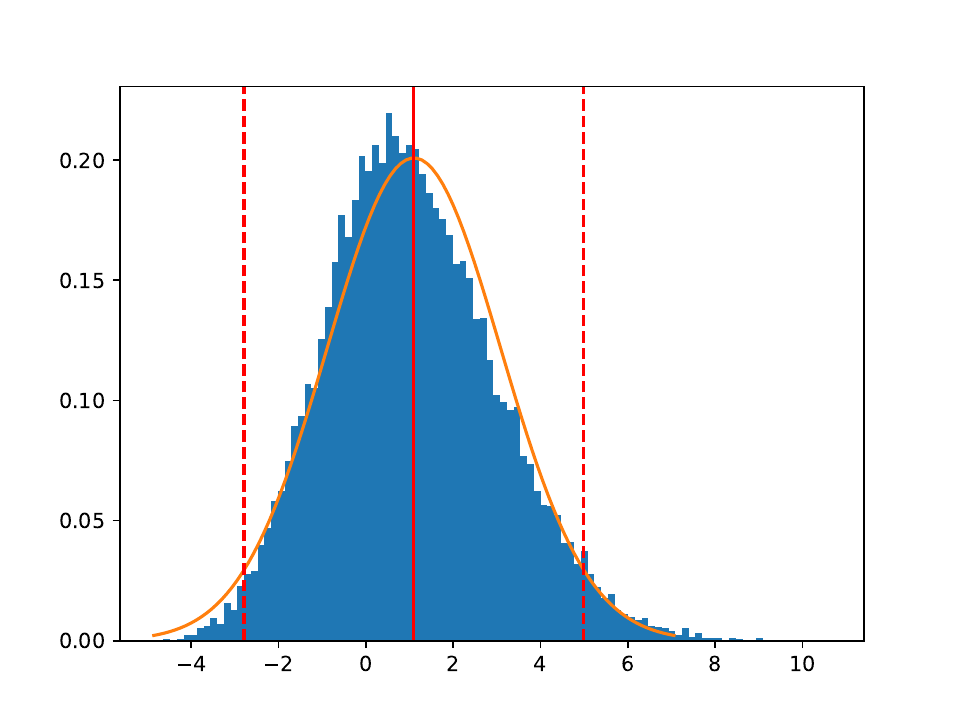}
         \caption{$T=200$}
         \label{fig:circular-T-200}
     \end{subfigure}
     \hfill
     \begin{subfigure}[b]{0.32\textwidth}
         \centering
         \includegraphics[width=\textwidth]{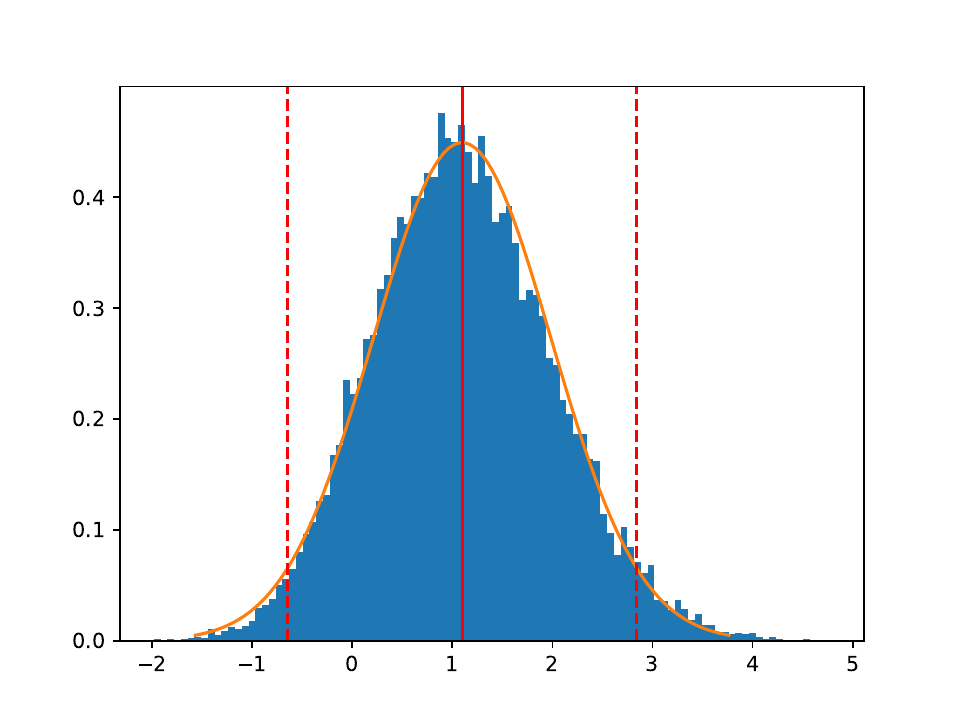}
         \caption{$T=1000$}
         \label{fig:circular-T-1000}
     \end{subfigure}
     \hfill
     \begin{subfigure}[b]{0.32\textwidth}
         \centering
         \includegraphics[width=\textwidth]{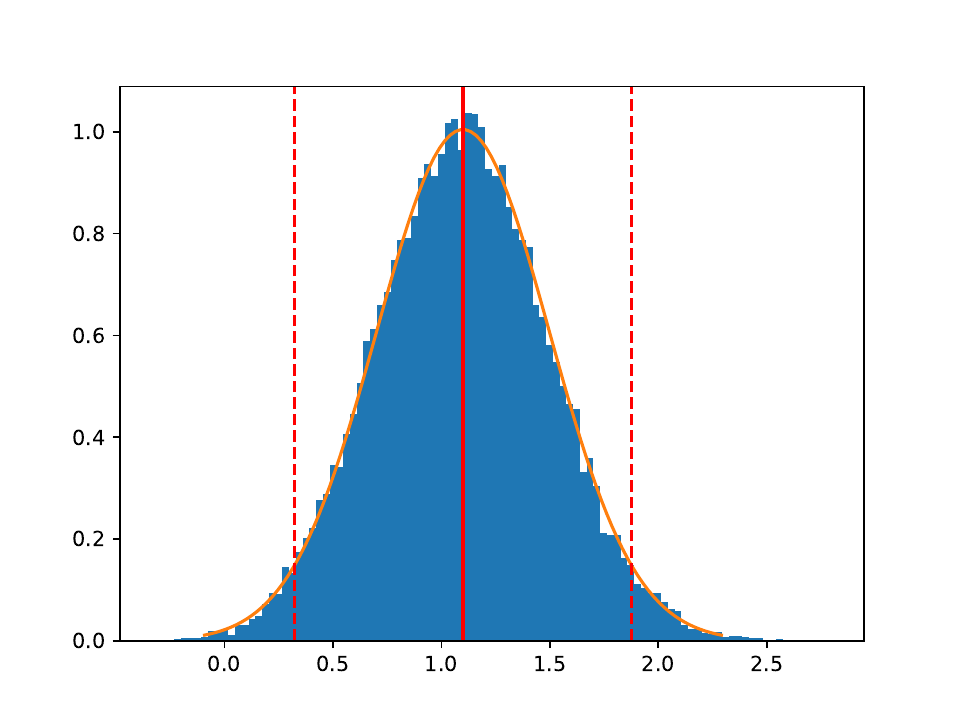}
         \caption{$T=5000$}
         \label{fig:circular-T-5000}
     \end{subfigure}
        \caption{Circular convolution model $\by = \bx \ast g$: we plot the histogram of $\widehat{\tau}(\mathbf{1}_{<K})$ defined in \eqref{eqn:estimator-lag-K} based on 20000 Monte Carlo simulations, the red solid vertical line corresponds to $\tau(\mathbf{1}_{<K})$ defined in \eqref{eqn:lag-K}, the red dashed lines are two standard deviations from the expectation using the variance formula in Proposition~\ref{prop:second-moment}. }
        \label{fig:circular-convolution}
\end{figure}

\begin{figure}
     \centering
     \begin{subfigure}[b]{0.32\textwidth}
         \centering
         \includegraphics[width=\textwidth]{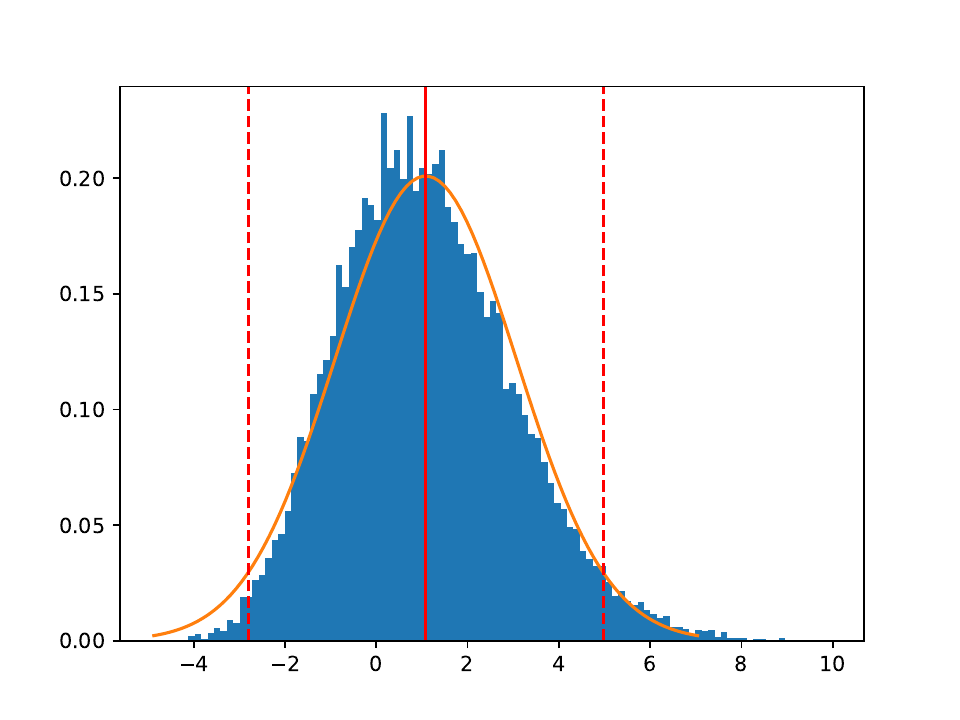}
         \caption{$T=200$}
         \label{fig:linear-T-200}
     \end{subfigure}
     \hfill
     \begin{subfigure}[b]{0.32\textwidth}
         \centering
         \includegraphics[width=\textwidth]{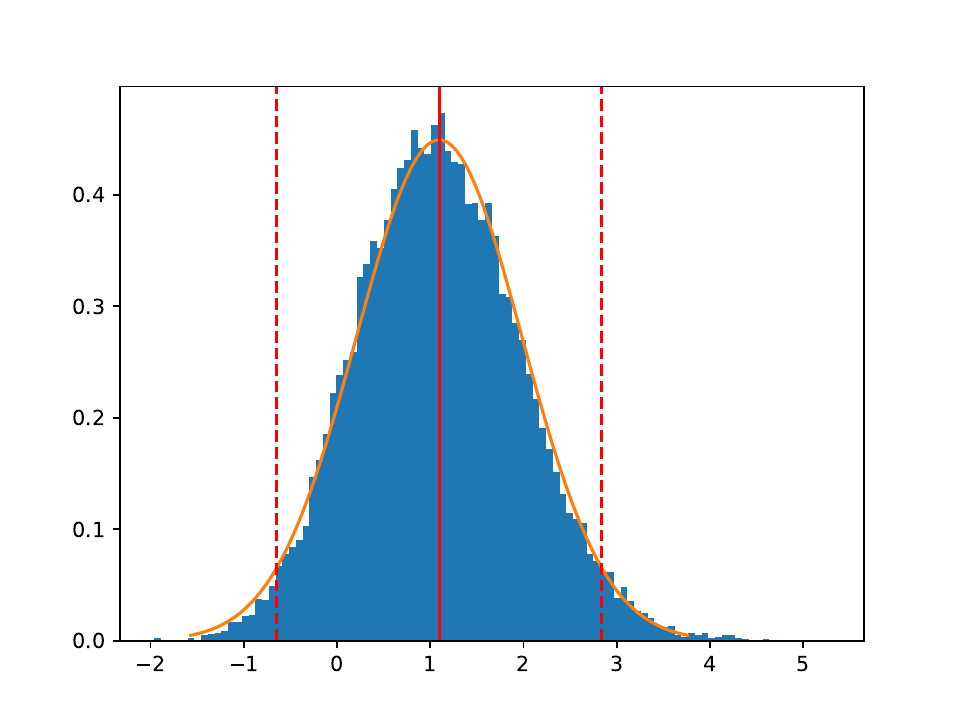}
         \caption{$T=1000$}
         \label{fig:linear-T-1000}
     \end{subfigure}
     \hfill
     \begin{subfigure}[b]{0.32\textwidth}
         \centering
         \includegraphics[width=\textwidth]{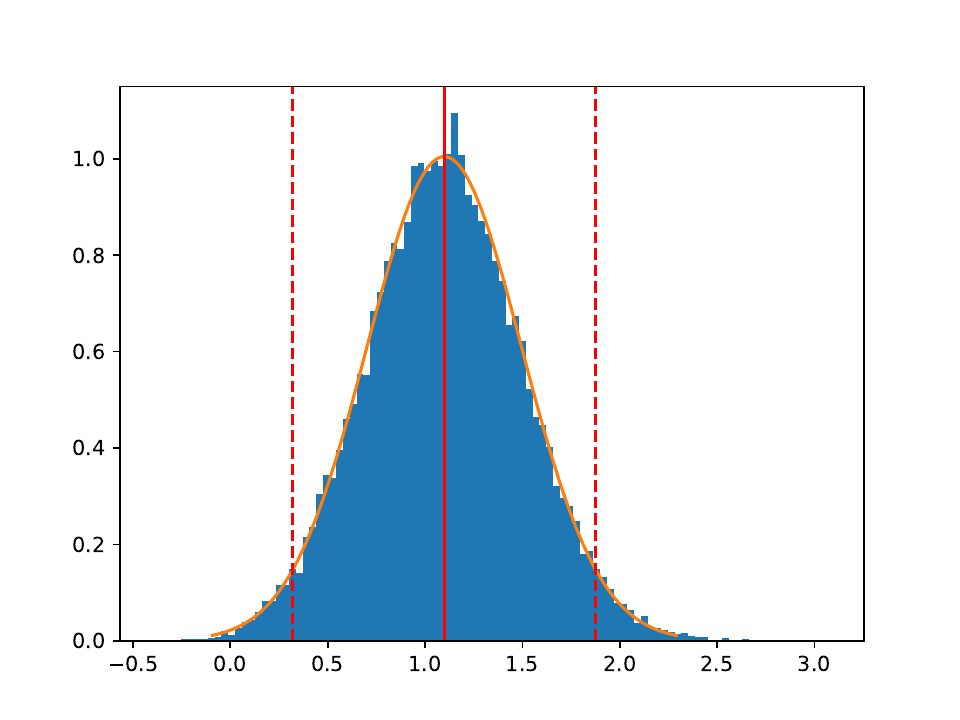}
         \caption{$T=5000$}
         \label{fig:linear-T-5000}
     \end{subfigure}
        \caption{Linear convolution model $\by = \bx * g$: we plot the histogram of $\widehat{\tau}^{\mathcal{L}}(\mathbf{1}_{<K})$ defined in \eqref{eqn:estimator-lin} based on 20000 Monte Carlo simulations, the red solid vertical line corresponds to $\tau^{\mathcal{L}}(\mathbf{1}_{<K})$ defined in \eqref{eqn:estimand-lin}, the red dashed lines are two standard deviations from the expectation using the variance formula in Proposition~\ref{prop:second-moment}.}
        \label{fig:linear-convolution}
\end{figure}

A few observations follow from the simulations in Fig.~\ref{fig:circular-convolution} and ~\ref{fig:linear-convolution}. First, the linear convolution and circular convolution are asymptotically equivalent, as shown in Section~\ref{sec:circular_vs_linear}. Here for each $T$, the distributions of $\widehat{\tau}(\mathbf{1}_{<K})$ and $\widehat{\tau}^{\mathcal{L}}(\mathbf{1}_{<K})$ are indiscernible, especially for large $T$. Second, the variance formulae in circular convolution, as shown in Proposition~\ref{prop:second-moment}, correctly provide the $95\%$ coverage even in the linear convolution model, seen in Fig.~\ref{fig:linear-convolution}. Even for $T=200$, Fig.~\ref{fig:linear-T-200} demonstrate a $95.49\%$ coverage, implying that the variance formula we derive under circular convolution remains an accurate approximation for the linear convolution model. Third, for each Figure, as $T$ increases, the distributions approach an asymptotic normal distribution, validating the asymptotic normality results established in Section~\ref{sec:asmptotic-normality}.

\section{Discussion and Future Work}

Formalizing a statistical framework for N-of-1 designs using dynamical systems provides a convenient modeling test-bed for within-subject analyses. Our initial results here suggest that this project is feasible, and that it is relatively simple to write software to compute treatment effects under relatively simple dynamic interference patterns. However, it is mildly discouraging that even the most simple linear time-invariant model is already quite complicated. We hope that our initial work here encourages investigation into simpler designs with simpler analyses. Or, perhaps, theoretical arguments will be made as to why such simplicity is impossible.

As this paper is just a first step towards such a theory of N-of-1 designs, there are plenty of open questions that must be answered to solidify the statistical foundations. First, in this work, we only consider the simplest randomization design, toggling treatment randomly at each time step. An analysis of balanced designs would likely immediately reduce variance in simple settings. Moreover, more clever experiment designs could potentially mitigate the long-term dependencies inherent to the statistical model. Such designs could potentially reveal more information about the system from shorter observation times. 

Another avenue for future work would be determining the optimal timing of interventions. Our statistical models are continuous processes sampled at a particular cadence of uniform discrete intervals. What is the best choice of this sampling time in order to most rapidly infer treatment effects? For example, in the case of medications, there will be time associated with the onset of the effect of treatment and time associated with the decay of the effect. Sampling too early or late will understate the impact of the treatment. In order to maximize the signal in an experiment, how can we determine the optimal intervention time? Can novel randomization schemes or adaptive designs help identify the optimal timing?

Finally, we hope that future work will investigate extensions beyond the linear time-invariant model studied in this work. We expect it would require little effort to extend the time-invariant treatment effects of this paper to study those where the response varies at each time step
$$
	y_t = \sum_{s=0}^{t-1} g_s(t-s) x_s\,.
$$
This would make our framework a full generalization of SUTVA potential outcomes. That is, with this particular time-variant model, the SUTVA model is recovered when $g_k(t)=0$ for all $t>0$. The main stumbling block is how to generalize the estimate of the variance of the treatment effect for this model. Extensions to nonlinear models would also add to the generality of this framework. It could be useful to appeal to the control literature to understand how to model nonlinear dose responses and saturation effects between doses. Future work will benefit from such further synergies between causal inference and mathematical control.

\section*{Acknowledgements}
The authors would like to thank Peng Ding, Avi Feller, Max Simchowitz, Panos Toulis, and Ruey Tsay for helpful suggestions and pointers to related work. This work was supported in part by ONR Award N00014-20-1-2497, NSF CCF CIF Award 2326498, and NSF Career Award (DMS-2042473). Additionally, TL acknowledges the generous support of the William Ladany Faculty Fellowship from the University of Chicago Booth School of Business and BR the generous support of the Miller Institute for Basic Research in Science.

%%%%%%%%%%%%%%%%%%%%%%%%%%%%%%%%%%%%%%%%%%%%%%
%% Reference                                %%
%%%%%%%%%%%%%%%%%%%%%%%%%%%%%%%%%%%%%%%%%%%%%%

{\footnotesize
\bibliography{ref.bib} 
\bibliographystyle{plainnat}
}

\appendix

\section{Technical Proofs}

\subsection{Moments of Rademacher Chaos}\label{app:chaos}

\begin{proposition}
	\label{prop:edge-indices-graph}
Let $W$ be a $d\times d$ symmetric matrix and let $\bz$ be a vector of $d$ independent Rademacher random variables. Consider the random variable $\bG :=  \sum_{i<j} \bz_i \bz_j W_{ij}$.
Then 
	\begin{align*}
		\E[\bG^2] = \sum_{i<j} W^2_{ij}
	\end{align*}
	and
	\begin{align*}
		\E [\bG^4] = 3 \big( \E [\bG^2] \big)^2 -  \sum_{i \neq j} W^4_{ij} + 3  \sum_{i \neq j \neq k \neq l} W_{ij} W_{jk} W_{kl} W_{li} \;.
	\end{align*}
	
\end{proposition}
\begin{proof}[Proof of Proposition~\ref{prop:edge-indices-graph}]
	We can define a graph $\cG = (\cV, \cE)$ on $|\cV| = T$ indices, and the edge weight $W_{e}:=W_{ij}, ~e \in \cE$ denotes an undirected weight between nodes $i\leftrightarrow j$. Then
	\begin{align*}
		\bG :=  \sum_{e \in \cE} \bz_{e} W_{e}
	\end{align*}
	where for an edge $e: i\leftrightarrow j$, $\bz_{e}:= \bz_{i}\bz_{j}$. Note that $\E[\bz_{e} \bz_{e'}]=0$ if $e\neq e'$. 	With this edge index notation, the second moment can be written as
	\begin{align*}
		\E [\bG^2] = \sum_{e \in \cE} W_{e}^2 + \sum_{e_1 \neq e_2 \in \cE} \E[z_{e_1} z_{e_2}]  W_{e_1}  W_{e_2} = \sum_{e \in \cE} W_{e}^2 \;,
\end{align*}

The edge indices prove convenient in evaluating the fourth moment. There are only three cases where $\E[\bz_{e_1} \bz_{e_2}\bz_{e_3}\bz_{e_4}]$ can be nonzero. First, if all of these edges are equal. Second, if there are two pairs of equal edges. Third, if there are four distinct edges that form a cycle.

Thus we have
\begin{align*}
		\E [\bG^4] = \sum_{e \in \cE} W_{e}^4 +  \frac{\binom{4}{2} \binom{2}{2}}{2!}\sum_{e_1 \neq e_2 \in \cE}   W_{e_1}^2  W_{e_2}^2 + \sum_{e_1 \neq e_2 \neq e_3 \neq e_4 \in \cE:~ \text{cycle}} W_{e_1}  W_{e_2}  W_{e_3}  W_{e_4} \;.
	\end{align*}
	
The sum of the first two terms equals $\E [\bG^4] = 3 \big( \E [\bG^2] \big)^2 -  \sum_{i \neq j} W^4_{ij}$. The last term is equal to 
\begin{align*}
	 3  \sum_{i \neq j \neq k \neq l} W_{ij} W_{jk} W_{kl} W_{li} 
\end{align*}
To derive this, fix four distinct points $i,j,k,l$ (say $i=1,j=2,k=3,l=4$). Then then there will be $ \frac{4!}{4*2} = 3$ distinguishable edge arrangements ($1\leftrightarrow 2 \leftrightarrow 3 \leftrightarrow 4 \leftrightarrow 1$, $1 \leftrightarrow 2 \leftrightarrow 4 \leftrightarrow 3 \leftrightarrow 1$, $1 \leftrightarrow 4 \leftrightarrow 2 \leftrightarrow 3 \leftrightarrow 1$) that form a cycle. In total, there will be $3 \times 4! = 72$ terms that involve node indices $1,2,3,4$, which can be grouped into
	\begin{align*}
		 \sum_{\substack{e_1 \neq e_2 \neq e_3 \neq e_4 \in \cE:\\~ \text{cycle with indices \{1,2,3,4\}}}} W_{e_1}  W_{e_2}  W_{e_3}  W_{e_4}  =24  W_{12} W_{23} W_{34} W_{41} + 24  W_{12} W_{24} W_{43} W_{31} + 24  W_{14} W_{42} W_{23} W_{31} \;.
	\end{align*}
On the other hand, in the summation over indices, there will be $4!=24$ terms with the index set $\{1,2,3,4\}$:	
	\begin{align*}
		\sum_{\substack{i \neq j \neq k \neq l:\\~ \{i,j,k,l\} = \{1,2,3,4\}}} W_{ij} W_{jk} W_{kl} W_{li} = 8  W_{12} W_{23} W_{34} W_{41} + 8  W_{12} W_{24} W_{43} W_{31} + 8  W_{14} W_{42} W_{23} W_{31} \;.
	\end{align*}
This verifies the desired moment formula.

%	The last term means that the subgraph formed by edges $(e_1, e_2, e_3, e_4)$ must form a cycle,
%	\begin{align}
%		\sum_{e_1 \neq e_2 \neq e_3 \neq e_4 \in \cE:~ \text{cycle}} W_{e_1}  W_{e_2}  W_{e_3}  W_{e_4} \;
%	\end{align}
%	and 

	%
	% by the symmetry of $w_{ij} = w_{ji}$ can be rearranged to be of the form $W_{ij} W_{jk} W_{kl} W_{li} $.
	
%    	Therefore, we have
%    	\begin{align}
%    		\E [\bG^4] = 3 \big( \E [\bG^2] \big)^2 -  \sum_{i \neq j} W^4_{ij} + 3  \sum_{i \neq j \neq k \neq l} W_{ij} W_{jk} W_{kl} W_{li} \;.
%    	\end{align}
\end{proof}

The following Lemma specializes these moment calculations to a particular form of Rademacher chaos that appears several times in this work.

\begin{lemma}\label{lemma:chaos}
Let $M$ be a matrix and $\bz$ be a vector of independent. Rademacher random variables. Then 
$$
	\E[ \langle \bz, M \bz \rangle] = \trace(M)
$$
and
$$
	\E[(\langle \bz, M \bz \rangle-\trace(M))^2]
		= \|M\|_F^2 + \trace(M^2) - 2 \sum_{i} M_{ii}^2 
$$	
\end{lemma}

\begin{proof}[of Lemma~\ref{lemma:chaos}]
The first moment is an elementary calculation. 
For the second moment, we can write
$$
	\langle \bz, M \bz \rangle-\trace(M) = \sum_{i<j} z_i z_j (M_{ij}+M_{ji})  \,.
$$
Apply Proposition~\ref{prop:edge-indices-graph} to the matrix $W=M+M^\transp$ to complete the proof.
\end{proof}

\subsection{Properties of Circular Convolution}

We collect a key fact about circular convolution that will be repeatedly used.

\begin{proposition}[Circular Symmetry]
	\label{prop:circular-symmetry}
	For any two sequence $a, b \in \mathbb{R}^T$, then for any $t, s \in [T]$
	\begin{align*}
		\sum_{i\in [T]} a^{\circ}_{t-i} b^{\circ}_{s-i} = \sum_{i \in [T]} a^{\circ}_{i-s} b^{\circ}_{i-t}
	\end{align*}
\end{proposition}
\begin{proof}[Proof of Proposition~\ref{prop:circular-symmetry}]
	Define the correspondence $\sigma:~ i \mapsto t+s-i\pmod{T}$, note this is a bijective map from $\sigma: [T]\rightarrow [T]$.
	\begin{align*}
		\sum_{i\in [T]} a^{\circ}_{t-i} b^{\circ}_{s-i} = \sum_{i \in [T]} a^{\circ}_{t-\sigma(i)} b^{\circ}_{s-\sigma(i)} = \sum_{i \in [T]} a^{\circ}_{i-s} b^{\circ}_{i-t} \;.
	\end{align*}
\end{proof}

\subsection{Proof of Proposition~\ref{prop:second-moment}}

	Observe that
	\begin{align*}
		\E_{\bx} \left[ \big( \widehat{\Delta}(q) -  \Delta(q) \big)^2 \right]  = \E[ \bW_T^2 ] = \tfrac{1}{T^2} \sum_{i,j} \E[\bw_i \bw_j] \;.
	\end{align*}
	Note for $i\neq j$ (w.l.o.g. say $i>j$), due to the martingale difference property
	\begin{align*}
		\E[\bw_i \bw_j] = \E\big[ \bw_j \E [ \bw_i | \cF_j ]\big] = 0 \;.
	\end{align*}
	Therefore we need to calculate 
	\begin{align*}
		\E[ \bW_T^2 ] &= \tfrac{1}{T^2} \sum_{i}\E[\bw_i^2] = \tfrac{1}{T^2} \sum_i  \sum_{j:j<i} (H_{ij}+H_{ji})^2 + L_i^2 \\
		&= \tfrac{1}{T^2} \sum_{i\neq j}  (H_{ij}^2 + H_{ij}H_{ji}) + \frac{1}{T^2} \sum_i L_i^2 \;.
	\end{align*}
	The remaining argument is algebraic manipulation using properties of circular convolution. We will show
	\begin{align*}
		\sum_{i\neq j}  H_{ij}^2 + H_{ij}H_{ji} &= T \cdot \left(\| g \ast q \|_2^2 + \langle g \ast g, q \ast q \rangle - 2 \langle q, g \rangle^2 \right) \;,\\
		\sum_i L_i^2 &=  \| (\mathbf{1} \ast g + 2e) \ast  q \|_2^2\;.
	\end{align*}
	
First, for $H$ we have
	\begin{align*}
		\sum_{i, j} H_{ij}^2  &= \sum_{i, j}  \sum_{t \in [T]} q^{\circ}_{t-i} g^{\circ}_{t-j} \sum_{s \in [T]} q^{\circ}_{s-i} g^{\circ}_{s-j} \\
		&= \sum_{t,s} \sum_{i}  q^{\circ}_{t-i} q^{\circ}_{s-i} \big(\sum_{j} g^{\circ}_{t-j} g^{\circ}_{s-j}\big) \\
		&= \sum_{t,s} \sum_{i}  q^{\circ}_{t-i} q^{\circ}_{s-i} \big(\sum_{j} g^{\circ}_{j-s} g^{\circ}_{j-t}\big) \quad \text{by Prop.~\ref{prop:circular-symmetry}}  \\
		&= \sum_{i,j} \big(\sum_{t} g^{\circ}_{j-t}  q^{\circ}_{t-i} \big) \big(\sum_{s} g^{\circ}_{j-s}  q^{\circ}_{s-i} \big)\\
		&= \sum_{i,j} (g \ast q)_{j-i\pmod{T}}^2 \\
		&= T \cdot \| g \ast q \|_2^2 \;.
	\end{align*}
	Similarly, we have
	\begin{align*}
		\sum_{i, j} H_{ij} H_{ji} &= \sum_{i, j}  \sum_{t \in [T]} q^{\circ}_{t-i} g^{\circ}_{t-j} \sum_{s \in [T]} q^{\circ}_{s-j} g^{\circ}_{s-i} \\
		&= \sum_{t,s} \sum_{i}  q^{\circ}_{t-i} g^{\circ}_{s-i} \big(\sum_{j} g^{\circ}_{t-j} q^{\circ}_{s-j}\big) \\
		& =\sum_{t,s} \sum_{i}  q^{\circ}_{t-i} g^{\circ}_{s-i} \big(\sum_{j} g^{\circ}_{j-s} q^{\circ}_{j-t}\big) \quad \text{by Prop.~\ref{prop:circular-symmetry}} \\
		& = \sum_{i,j} \big(\sum_{t} q^{\circ}_{j-t}  q^{\circ}_{t-i} \big) \big(\sum_{s} g^{\circ}_{j-s}  g^{\circ}_{s-i} \big) \\
		& = \sum_{i,j} (g \ast g)_{j-i\pmod{T}} (q \ast q)_{j-i\pmod{T}}  \\
		&= T \cdot \langle g \ast g, q \ast q \rangle \;.
	\end{align*}
	And thus
	\begin{align*}
		 \sum_{i\neq j} H_{ij}^2 + H_{ij}H_{ji} =  T \cdot \left(\| g \ast q \|_2^2 + \langle g \ast g, q \ast q \rangle - 2 \langle q, g \rangle^2 \right) \;.
	\end{align*}
	
	For $L$, let $h=\mathbf{1} \ast g + 2e$. We then have
	\begin{align*}
		\sum_i L_i^2 &= \sum_i \sum_{t,s}  q^{\circ}_{t-i} h_{t} q^{\circ}_{s-i} h_{s} = \sum_{t,s} h_{t} h_{s} \sum_i  q^{\circ}_{t-i} q^{\circ}_{s-i} \\
		&= \sum_{t,s} h_{t} h_{s}  \sum_i  q^{\circ}_{i-s} q^{\circ}_{i-t} =  \sum_i (h \ast q)_i^2 = \| h \ast q \|_2^2 \;.
	\end{align*}

\subsection{Proof of Proposition~\ref{prop:plug-in-estimate}}
\begin{proof}[Proof of Proposition~\ref{prop:plug-in-estimate}]
To show the consistency of $\widehat{\cV}_{Q} \rightarrow \cV_{Q}$, we will show consistency for each term in \eqref{eqn:est-VQ}.

	First we will show $\| \widehat{g}_{<K} \ast q \|_2^2 \stackrel{a.s.}{\rightarrow} \| g \ast q \|_2^2$. For $t\geq 2K$
	\begin{align*}
		|(g \ast q)_t| \leq  \sum_{i \in [K]} |g_{t-i}| \leq K c^{t-K}
	\end{align*}
	and thus $\| g \ast q \|_2^2 - \sum_{t \in [2K]} (g \ast q)_t^2 \leq K^2 c^{2K} = o_T(1)$ with $K = \Theta(\log(T))$. Note that for $t < 2K$,
	$
		(g \ast q)_t = \sum_{i \in [K]} g^{\circ}_{t-i} q_i
	$
	and recall the uniform consistency for $\sup_{k < 2K} \left| \widehat{\tau}(u_k) - g_k \right| \precsim \tfrac{\log(T)}{\sqrt{T}}$ (Theorem~\ref{thm:uniform-consistency}), we have with probability at least $1- 4T^{-2}$
	\begin{align*}
		| (g \ast q)_t - (\widehat{g}_{<K} \ast q)_t | \precsim \left\{
		\begin{matrix}
			K\tfrac{\log(T)}{\sqrt{T}} + c^{T-1-K} = o_T(1) & \text{for}~ t \in [0, K)\\
			K\tfrac{\log(T)}{\sqrt{T}} =  o_T(1) & \text{for}~ t \in [K, 2K)
		\end{matrix} \right.
	\end{align*}
	and thus, we have
	\begin{align*}
		\| \widehat{g}_{<K} \ast q \|_2^2 - \sum_{t\in[2K]} (g \ast q)_t^2 = o_T(1) \;.
	\end{align*}
	To sum up, the above derivation implies $\| \widehat{g}_{<K} \ast q \|_2^2 \stackrel{a.s.}{\rightarrow} \| g \ast q \|_2^2$. It is immediate to see that the proof follows the same steps if one changes the circular convolution to the linear convolution; one needs to change the subscript from $i\in [K]$ to $i \in [K \wedge t]$, and from $g_{t-i}^\circ$ to $g_{t-i}$.

	Now we study the term $\langle  \widehat{g}_{<2K} \ast  \widehat{g}_{<2K}, q \ast q \rangle \stackrel{a.s.}{\rightarrow}  \langle g \ast g, q \ast q \rangle$. First note that
	\begin{align*}
		\langle g \ast g, q \ast q \rangle = \sum_{t \in [2K]} (g \ast g)_t (q \ast q)_t \;,
	\end{align*}
	and that $q \ast q$ is supported on the first $2K$-entries.
	For $t < 2K$, $(g \ast g)_t - \sum_{i \in [2K]} g_i g^{\circ}_{t-i} \precsim c^{2K}$, and by the uniform consistency in Theorem~\ref{thm:uniform-consistency} and triangle inequality, we have
	\begin{align*}
		&| (g \ast g)_t - (\widehat{g}_{<2K} \ast  \widehat{g}_{<2K})_t |  \leq | \sum_{i \in [2K]} g^{\circ}_{t-i}  g_i - \sum_{i \in [2K]} \widehat{g}^{\circ}_{t-i} \widehat{g}_{i} |  + c^{2K} \\
		& \leq \sum_{i \in [2K]} g^{\circ}_{t-i} \cdot \max_{i\in [2K]}|\widehat{g}_i - g_i| + \sum_{i \in [2K]} |\widehat{g}_{i}| \cdot \max_{i\in [2K]} |\widehat{g}^{\circ}_{t-i} - g^\circ_{t-i} | + c^{2K}  \precsim
		    \mathrm{Poly}(K) \tfrac{\log(T)}{\sqrt{T}}	+ c^{2K}  \;. \nonumber
	\end{align*}
	Therefore,
	$\langle  \widehat{g}_{<2K} \ast  \widehat{g}_{<2K}, q \ast q \rangle \stackrel{a.s.}{\rightarrow}  \langle g \ast g, q \ast q \rangle$. Again, the proof follows the same steps if one changes the circular convolution to the linear convolution.

	Lastly, Theorem~\ref{thm:uniform-consistency} showed that
	\begin{align*}
		 \langle  \widehat{g}_{K}, q \rangle -  \langle  g, q \rangle   \stackrel{a.s.}{\rightarrow} 0 \;.
	\end{align*}
	Put three terms together, we have shown $\widehat{\cV}_{Q}  \stackrel{a.s.}{\rightarrow} \cV_{Q}$.

	Now we move on to study the term $\cV_{L}$.
	\begin{align*}
		\cV_{L} &=\tfrac{1}{T}\| (\langle \mathbf{1},g \rangle \mathbf{1} + 2e) \ast  q \|_2^2 = \tfrac{1}{T} \sum_{t\in [T]} \big(  \langle \mathbf{1},g \rangle \langle \mathbf{1},q \rangle  + 2 \sum_{i\in [K]} e_{t-i}^\circ  q_i \big)^2
	\end{align*}
	We introduce the estimator for the error vector
	\begin{align*}
		\widehat{e}_t := \by_t - (\bx \ast \widehat{g}_{<K})_t = e_t +  \underbrace{(\bx \ast g_{<K})_t -  (\bx \ast \widehat{g}_{<K})_t}_{(i)} + \underbrace{(\bx \ast g_{\geq K})_t }_{(ii)} \;.
	\end{align*}
	To bound the term $(i)$, we note that by Theorem~\ref{thm:uniform-consistency}
	\begin{align*}
		| (\bx \ast g_{<K})_t -  (\bx \ast \widehat{g}_{<K})_t | = | \sum_{i \in [K]} \bx^\circ_{t-i} (g_i - \widehat{g}_i) | \leq K \sup_{i \in [K]} |g_i - \widehat{g}_i| \precsim \frac{K\log(T)}{\sqrt{T}} \;.
	\end{align*}
	To bound the term $(ii)$, we use the fact
	\begin{align*}
		| (\bx \ast g_{\geq K})_t | \leq \| g_{\geq K} \|_1 \| \bx \|_{\infty} \precsim c^{K} \;.
	\end{align*}
	Thus for all $t \in [T]$,
	\begin{align*}
		\sup_{t \in [T]} | \widehat{e}_t - e_t |  &\precsim \frac{K\log(T)}{\sqrt{T}} + c^{K} \;, \\
		\sup_{t \in [T]} | \sum_{i\in [K]} \widehat{e}_{t-i}^\circ q_i - \sum_{i\in [K]} e_{t-i}^\circ q_i |   &\precsim \frac{K^2\log(T)}{\sqrt{T}} + Kc^{K}  = o_T(1) \;.
	\end{align*}
	So far, we have proved
	\begin{align*}
		\sup_{t \in [T]} \left| \big( \langle \mathbf{1}, \widehat{g}_K \rangle \langle \mathbf{1},q \rangle +  2 (\widehat{e} \ast  q)_t\big)  -   \big(\langle \mathbf{1},g \rangle \langle \mathbf{1},q \rangle + 2 (e \ast  q)_t \big) \right| = o_T(1)
	\end{align*}
	and thus, as a direct consequence $\widehat{\cV}_{L}  \stackrel{a.s.}{\rightarrow} \cV_{L}$.
\end{proof}

\subsection{Proof of Lemma~\ref{lemma:H-inf-norm}}

By the circular symmetry property in Proposition~\ref{prop:circular-symmetry}, we have
	\begin{align*}
\max_{j: j\neq i} \omega_{ij}^2 \leq
		\sum_{j} \omega_{ij}^2 - \omega_{ii}^2 &= 2\sum_{t,s}  q^{\circ}_{t-i}  q^{\circ}_{s-i} \sum_j g^{\circ}_{t-j}  g^{\circ}_{s-j} + 2  \sum_{t,s}  q^{\circ}_{t-i}  g^{\circ}_{s-i} \sum_j g^{\circ}_{t-j}  q^{\circ}_{s-j} - 4 \langle g, q\rangle^2  \\
		& = 2 \sum_{t,s}  q^{\circ}_{t-i}  q^{\circ}_{s-i} \sum_j g^{\circ}_{j-s}  g^{\circ}_{j-t} + 2  \sum_{t,s}  q^{\circ}_{t-i}  g^{\circ}_{s-i} \sum_j g^{\circ}_{j-s}  q^{\circ}_{j-t} - 4 \langle g, q\rangle^2  \\
		&  = 2\sum_j (g \ast q)^2_{j-i\pmod{T}} + 2\sum_j (g \ast g)_{j-i\pmod{T}} (q \ast q)_{j-i\pmod{T}} - 4 \langle g, q\rangle^2  \\
		& = 2\| g \ast q \|_2^2 + 2 \langle g \ast g, q \ast q \rangle - 4 \langle g, q\rangle^2  \;.
	\end{align*}
Note that the above holds for any $i$, we finish the proof.

\subsection{Proof of Lemma~\ref{lemma:fourth-moment}}

Define $\omega_{ij} := H_{ij}+ H_{ji}$. 
we have
	\begin{align*}
		\omega_{ij}  = \sum_{t \in [T]} q^\circ_{t-i} g^\circ_{t-j} + q^\circ_{t-j} g^\circ_{t-i} = \omega_{ji} \;.
	\end{align*}
Let $\bG = \sum_{i < j } \bz_i \bz_j \omega_{ij} = T \cdot \bH_{T}$. Invoking Proposition~\ref{prop:edge-indices-graph} yields the following equality
 	\begin{align}\label{eq:desired_quartic_expression}
 		\frac{\E [\bH_{T}^4 ]}{\big( \E [\bH_T^2] \big)^2 } = 3 -  \frac{\sum_{i \neq j} \omega^4_{ij}}{T^4 \big( \E [\bH_T^2] \big)^2} + \frac{\sum_{i \neq j \neq k \neq l} \omega_{ij} \omega_{jk} \omega_{kl} \omega_{li}}{T^4 \big( \E [\bH_T^2] \big)^2} \; .
 	\end{align}	
	
We already have by the argument in Proposition~\ref{prop:second-moment} that
\begin{align*}
	\E [\bH_T^2] = \tfrac{1}{T} \cdot \left(\| g \ast q \|_2^2 + \langle g \ast g, q \ast q \rangle - 2 \langle q, g \rangle^2 \right)\,.
\end{align*}
Thus, to complete the theorem we need to bound the quantites
\begin{align}\label{eq:quartic_terms_to_bound}
\sum_{i \neq j} \omega^4_{ij}~~\text{and}~~\sum_{i \neq j \neq k \neq l} \omega_{ij} \omega_{jk} \omega_{kl} \omega_{li}\,.
\end{align}

For the first term, we will show
	\begin{align}\label{eq:quartic_ineq1}
		\sum_{i \neq j} \omega_{ij}^4 \leq 4T  \cdot \left(\| g \ast q \|_2^2 + \langle g \ast g, q \ast q \rangle - 2 \langle q, g \rangle^2 \right)^2 \;,
	\end{align}
To see this, we first establish
	\begin{align*}
		\sum_{i \neq j} \omega_{ij}^2 = 2 \sum_{i\neq j} H_{ij}^2 + H_{ij}H_{ji} = 2 T \cdot \left(\| g \ast q \|_2^2 + \langle g \ast g, q \ast q \rangle - 2 \langle q, g \rangle^2 \right) \;.
	\end{align*}
By Lemma~\ref{lemma:H-inf-norm}, we further have
\begin{align*}
	\max_{i \neq j} \omega_{ij}^2 \leq 2\| g \ast q \|_2^2 + 2 \langle g \ast g, q \ast q \rangle - 4 \langle g, q\rangle^2  \;.
\end{align*}
Put these two bounds together, we find
	\begin{align*}
		\sum_{i \neq j} \omega_{ij}^4 &\leq \max_{i\neq j}  \omega_{ij}^2  \cdot \sum_{i \neq j} \omega_{ij}^2 \leq 4T  \left(\| g \ast q \|_2^2 + \langle g \ast g , q \ast q \rangle - 2 \langle q, g \rangle^2 \right)^2  \;,
	\end{align*}
which verifies the inequality~\eqref{eq:quartic_ineq1}.
	
Next, we aim to show
\begin{align}\label{eq:quartic_ineq2}
		|\sum_{i \neq j \neq k \neq l} \omega_{ij} \omega_{jk} \omega_{kl} \omega_{li} | \leq 16T \cdot\| (|g| \ast |q|) \ast (|g| \ast |q|) \|_2^2  \;.
\end{align}
To proceed, we first note
	\begin{align}
		&|\sum_{i \neq j \neq k \neq l} \omega_{ij} \omega_{jk} \omega_{kl} \omega_{li} | \leq \sum_{i, j,k, l} |\omega_{ij} \omega_{jk} \omega_{kl} \omega_{li} | \nonumber \\
		& \leq \sum_{i,  k}  \big(  \sum_{j}|\omega_{ij} \omega_{jk}| \big) \big(  \sum_{l} | \omega_{kl} \omega_{li} | \big) \label{eqn:crazy-sum}
	\end{align}
	
	Now, the above Equation~\eqref{eqn:crazy-sum} can be break down into a sum of $2^4 = 16$ terms
	\begin{align}
		 \sum_{j}|\omega_{ij} \omega_{jk}| & \leq \sum_{j}  \big(\sum_{t} |q^{\circ}_{t-i} g^{\circ}_{t-j}| + |q^{\circ}_{t-j} g^{\circ}_{t-i}| \big) \big( \sum_{s} |q^{\circ}_{s-j} g^{\circ}_{s-k}| + |q^{\circ}_{s-k} g^{\circ}_{s-j}| \big) \nonumber\\ 
		 & = \sum_{t,s} \sum_j |q^{\circ}_{t-i} g^{\circ}_{t-j}| |q^{\circ}_{s-j} g^{\circ}_{s-k}|  +  \sum_{t,s} \sum_j |q^{\circ}_{t-j} g^{\circ}_{t-i}|  |q^{\circ}_{s-k} g^{\circ}_{s-j}| \label{eqn:A-44}\\
		 & \quad + \sum_{t,s} \sum_j |q^{\circ}_{t-i} g^{\circ}_{t-j}||q^{\circ}_{s-k} g^{\circ}_{s-j}| + \sum_{t,s} \sum_j |q^{\circ}_{t-j} g^{\circ}_{t-i}| |q^{\circ}_{s-j} g^{\circ}_{s-k}| \label{eqn:A-45}
	\end{align}
	Each of the above terms is handled similarly, and we will only showcase one. Denote $|g|$ to be the vector taking entrywise absolute value to the vector $g$, we know
	\begin{align*}
		\sum_{t,s} \sum_j |q^{\circ}_{t-i} g^{\circ}_{t-j}| |q^{\circ}_{s-j} g^{\circ}_{s-k}| & = \sum_{t,s} \sum_j |q^{\circ}_{t-i} g^{\circ}_{j-s}| |q^{\circ}_{j-t} g^{\circ}_{s-k}| \\
		&= \sum_j (|q| \ast |q|)_{j-i}^\circ (|g| \ast |g|)_{j-k}^\circ
	\end{align*}
	where the first step is by the circular symmetry property in Proposition~\ref{prop:circular-symmetry}. 
	Using the same idea on all 4-terms in \eqref{eqn:A-44}-\eqref{eqn:A-45}, we have shown
	\begin{align*}
		 \sum_{j}|\omega_{ij} \omega_{jk}| & = \sum_j (|q| \ast |q|)_{j-i}^\circ (|g| \ast |g|)_{j-k}^\circ + \sum_j (|q| \ast |q|)_{j-k}^\circ (|g| \ast |g|)_{j-i}^\circ \\
		 & + \sum_j (|g| \ast |q|)_{j-i}^\circ (|g| \ast |q|)_{j-k}^\circ + \sum_j (|g| \ast |q|)_{j-k}^\circ (|g| \ast |q|)_{j-i}^\circ \;.
	\end{align*}
	
	Now one of the $16$-terms in \eqref{eqn:crazy-sum} takes the form
	\begin{align*}
		&\sum_{i,  k} \sum_j (|g| \ast |q|)_{j-i}^\circ (|g| \ast |q|)_{j-k}^\circ \sum_l (|g| \ast |q|)_{l-i}^\circ (|g| \ast |q|)_{l-k}^\circ  \\
		&= \sum_{i,  k} \sum_j (|g| \ast |q|)_{k-j}^\circ (|g| \ast |q|)_{i-j}^\circ \sum_l (|g| \ast |q|)_{l-i}^\circ (|g| \ast |q|)_{l-k}^\circ\\
		& = \sum_{j, l} \sum_i (|g| \ast |q|)_{i-j}^\circ (|g| \ast |q|)_{l-i}^\circ \sum_k (|q| \ast |q|)_{k-j}^\circ  (|q| \ast |q|)_{l-k}^\circ \\
		& = \sum_{j, l} ( (|g| \ast |q|) \ast (|g| \ast |q| ) )_{l-j}^\circ ( (|g| \ast |q| ) \ast (|g| \ast |q| ))_{l-j}^\circ \\
		& = T \langle (|g| \ast |q|) \ast (|g| \ast |q| )  , (|g| \ast |q|) \ast (|g| \ast |q| )  \rangle
	\end{align*}
	where the first step uses the circular symmetry property in Proposition~\ref{prop:circular-symmetry} again.
	Note that each of the $16$ terms will look similar to the above, which looks like an $8$-th moment with a balanced number of $g$ and $q$'s ($4$ each).

	To bound each of the $16$ terms by the representative term, we need to recall Parseval's identity and convolution theorem: for a vector $g \in \mathbb{R}^T$, denote $\reallywidehat{g} \in \mathbb{C}^T$ to denote its discrete-time Fourier transform (DFDT), then
	\begin{align*}
		\langle g, q \rangle = \tfrac{1}{T} \sum_{k \in [T]} \overline{\reallywidehat{g  }_{k}} \reallywidehat{q  }_{k}  \;.\\
		\reallywidehat{ g \ast q }_k = \reallywidehat{g  }_k \cdot \reallywidehat{q}_k \;.
 	\end{align*}
	Here, for a complex value $z \in \mathbb{C}$, we denote $\overline{z}$ to be its complex conjugate and $\| z \|$ to denote its modulus.
	Take one of the 16 terms, say
	\begin{align*}
		&\langle (|g| \ast |g|) \ast (|g| \ast |g|) , (|q| \ast |q|) \ast (|q| \ast |q| ) \rangle  \\
		& =  \tfrac{1}{T} \sum_{k \in [T]}  \overline{\reallywidehat{|g|}_k^4} \reallywidehat{|q|}_k^4  \\
		& \leq \tfrac{1}{T} \sum_{k \in [T]}  \big\| \reallywidehat{|g|  }_k \big\|^4 \big\|\reallywidehat{|q | }_{k}  \big\|^4 = \| (|g| \ast |q|) \ast (|g| \ast |q|) \|_2^2 \;.
	\end{align*}
Therefore, we have verified~\eqref{eq:quartic_ineq2}.
	
Using the inequalities~\eqref{eq:quartic_ineq1} and~\eqref{eq:quartic_ineq2} to upper bound~\eqref{eq:desired_quartic_expression} thus completes the proof.

\end{document}